\documentclass[journal]{IEEEtran}
 \ifCLASSINFOpdf
 \usepackage[pdftex]{graphicx}
\else
\fi
\usepackage{subcaption}
\usepackage[nodisplayskipstretch]{setspace}

\usepackage{amsmath}
\usepackage{amsthm}
\usepackage{amssymb}
\usepackage[shortlabels]{enumitem}
\usepackage{bbm}
\usepackage{scrextend}
\usepackage{mathtools}
\usepackage{extarrows}
\usepackage{url}
\usepackage{accents}
\usepackage{optidef}
\usepackage{array}
\usepackage{pdfpages}
\usepackage{titlesec}

\usepackage{multirow}
\newcommand{\ubar}[1]{\underaccent{\bar}{#1}}
\usepackage{cite}
\theoremstyle{definition}
\newtheorem{definition}{Definition}[section]
 
\newtheoremstyle{mytheorem}
  {1pt plus 1.0pt minus 1.0pt}
  {1pt plus 1.0pt minus 1.0pt}
  {\itshape}
  {}
  {\itshape\bfseries}
  {.}
  {.2em}
  {\thmname{#1}\thmnumber{ #2} \thmnote{ {\the\thm@notefont(#3)}}}

\makeatother
\theoremstyle{mytheorem} 
 
\newtheorem{theorem}{Theorem}
\newtheorem{corollary}{Corollary}[theorem]
\newtheorem{lemma}[theorem]{Lemma}
\newtheorem{proposition}[theorem]{Proposition}
\theoremstyle{remark}
\newtheorem*{remark}{Remark}

\makeatletter
\renewcommand*\env@matrix[1][*\c@MaxMatrixCols c]{%
  \hskip -\arraycolsep
  \let\@ifnextchar\new@ifnextchar
  \array{#1}}
\makeatother

\newcommand{\norm}[1]{\left\lVert#1\right\rVert}

\DeclarePairedDelimiter\abs{\lvert}{\rvert}%

\makeatletter
\let\oldabs\abs
\def\abs{\@ifstar{\oldabs}{\oldabs*}}

\makeatletter
\if@twocolumn
\newcommand{\whencolumns}[2]{
#2
}
\else
\newcommand{\whencolumns}[2]{
#1
}
\fi
\makeatother

\DeclareMathOperator{\spark}{spark}
\DeclarePairedDelimiter\ceil{\lceil}{\rceil}

\newcommand\blfootnote[1]{%
  \begingroup
  \renewcommand\thefootnote{}\footnote{#1}%
  \addtocounter{footnote}{-1}%
  \endgroup
}

\begin{document}
\title{\huge Source Coding Based Millimeter-Wave Channel Estimation with Deep Learning Based Decoding}


\author{\IEEEauthorblockN{Yahia Shabara \textit{Student Member, IEEE}, Eylem Ekici, \textit{Fellow, IEEE}, and C. Emre Koksal, \textit{Senior Member, IEEE}}}

\maketitle

\begin{abstract}
The \textit{speed} at which millimeter-Wave (mmWave) channel estimation can be carried out is critical for the adoption of mmWave technologies.
This is particularly crucial because mmWave transceivers are equipped with large antenna arrays to combat severe path losses, which consequently creates large channel matrices, whose estimation may incur significant overhead.
This paper focuses on the mmWave \textit{channel estimation} problem.
Our objective is to reduce the number of measurements required to reliably estimate the channel.
Specifically, channel estimation is posed as a ``\textit{source compression}'' problem in which measurements mimic an \textit{encoded} (compressed) version of the channel.
\textit{Decoding} the observed measurements, a task which is traditionally computationally intensive, is performed using a \textit{deep-learning-based} approach, facilitating a high-performance channel discovery.
Our solution not only outperforms state-of-the-art compressed sensing methods, but it also determines the lower bound on the number of measurements required for reliable channel discovery.
\end{abstract}
\begin{IEEEkeywords}
\whencolumns{\footnotesize}{}
Millimeter-Wave, Channel Estimation, Path Discovery, Sparse Recovery, Source Coding, Machine Learning.
\end{IEEEkeywords}
\blfootnote{This work was supported in part by the NSF CNS under Grant 1514260, Grant 1618566, Grant 1731698 and Grant 1814923.}

\IEEEpeerreviewmaketitle

\section{Introduction}
\IEEEPARstart{T}{he Rapid} increase in mobile data traffic has motivated the exploration of mmWave spectrum bands \cite{rappaport2013millimeter, pi2011introduction, akdeniz2014millimeter,cisco_2017}.
While mmWave communication promises orders of magnitude increase in data rates, it both i) suffers from severe path losses \cite{maccartney2013path} and ii) necessitates the use of power-hungry circuits to operate.
To overcome these problems, large-gain, highly-directional antenna arrays are proposed as a counter measure to path losses, along with less flexible, yet energy-efficient transceivers that no longer use fully-digital beamforming.
Large antenna arrays, however, create channel matrices with large dimensions, which are complex to estimate. When combined with limited transceiver capabilities, large scale channel estimation may take prohibitively long periods.
Reducing the number of measurements is thus a critical step towards facilitating mmWave networks. Fortunately, this does not necessarily degrade the quality of channel estimation due to the sparse nature of mmWave channels; a feature that has been revealed by empirical measurement studies and further adopted by statistical channel models \cite{akdeniz2014millimeter,rangan2014millimeter,haneda20165g}.

This work focuses on \textbf{\textit{the problem of mmWave channel estimation}} with the objective of decreasing the number of required measurements.
We treat this problem as that of path/beam discovery which is crucial for \textit{\textbf{initial link establishment}} between a transmitter (TX) and a receiver (RX) (also known as \textit{\textbf{Initial Access}}).
We solve this problem using a technique inspired by \textbf{\textit{binary source coding (data compression)}}.
Although binary codes are natively designed to compress binary data, we provide a foundation for the same codes to be used for compressing complex-valued data, as well.
We devise a method to obtain channel measurements such that they resemble a compressed version of the channel matrix.
To estimate the channel from the acquired measurements, we train a Deep Neural Network (DNN) that enables very high speed processing.
However, training DNNs that jointly process all measurements poses an overwhelming complexity.
Thus, we propose a novel computationally-tractable solution that sequentially processes the acquired measurements. This method is unique in the sense that it reduces the problem of estimating the \textit{channel matrix} as a whole into several smaller sub-problems of estimating the individual rows and columns of that matrix.
The key contributions of this work are as follows:
\begin{itemize}[noitemsep,topsep=0pt]
\item We show that lossless, fixed-rate, linear source codes can be used to design efficient channel measurements that can be \textit{uniquely} mapped to the underlying channels.
\item We accurately evaluate the number of measurements needed for reliable channel discovery (as opposed to a mere scaling law).
This number is dependent on the compression ratio of the chosen code.
\item We present a tight lower bound on the number of measurements needed to reliably discover the channel and provide a solution that achieves this bound.
\item We propose a high-performance DNN based measurement-to-channel mapping.
\item We show that our solution outperforms the state-of-the-art compressed sensing based solutions and the IEEE 802.11ad beam alignment method.
\end{itemize}

The mmWave channel estimation problem can generally be divided into two intertwined parts.
The first is: \textbf{\textit{how to obtain ``good'' measurements that can be used to reliably discover the channel?}}
and the second is: \textit{\textbf{how to map these measurements to corresponding channel estimates?}}
Motivated by our proposed solution, we name these two parts \textit{``Channel Encoding''} and \textit{``Measurement Decoding''}, respectively.
Encoding and decoding are intertwined because a selection of a specific decoding method often dictates (i) how the measurements are obtained, and (ii) the number of measurements for which this specific decoding method would yield ``good'' performance.
The dissociation of encoding and decoding as two sub-problems can be seen across almost all mmWave channel estimation research, albeit not always explicitly mentioned.
This distinction, however, facilitates the identification of key aspects upon which we could improve the quality of channel estimation.

A well-known classification of encoding paradigms is encoding \textit{with} vs. \textit{without} feedback.
Non-feedback encoding is better suited for simultaneous multi-user channel estimation, hence is scalable, while feedback-based encoding operates better at low SNR \cite{chiu2019active}.
Different decoding algorithms are also needed for these two types.
This paper focuses on encoding without feedback.

The state-of-the-art mmWave channel estimation algorithms rely on compressed Sensing (CS) to reduce the number of channel measurements\cite{CS4Wireless_TipsAndTricks, venugopal2017channel, alkhateeb2014channel,rodriguez2018frequency}.
Other approaches include: i) measurements with hierarchical beam patterns that sequentially narrow down the angular direction(s) which contain strong propagation paths, ii) measurements with overlapped beam patterns where each measurement combines signals received from a randomly selected set of angular directions \cite{hassanieh2018fast}, and iii) machine learning based algorithms for sparse recovery of mmWave channels \cite{li2019deep, he2018deep, huang2018deep, alkhateeb2018deep}.
Further, in \cite{shabaraBeamToN} we first introduced the idea of exploiting \textit{binary codes} for tackling mmWave channel estimation. In particular, we exploit the capability of error discovery of channel codes and construct an analogy to path discovery in mmWave channels.

\noindent \textit{Notations:} $x$ is a scalar quantity, while $\boldsymbol{x}$ is a vector and $\boldsymbol{X}$ is a matrix. The transpose of a matrix is denoted by $\boldsymbol{X}^T$, while $\boldsymbol{X}^*$ denotes its conjugate and $\boldsymbol{X}^H$ denotes the conjugate transpose.

\section{Related Work}
\textbf{Initial Access:} The ``Initial Access'' problem is concerned with finding the angular bearings of one or more propagation paths between a pair of \whencolumns{TX/RX}{TX and RX} nodes, without prior knowledge about previous channel values.
In mobile environments, these angular directions are expected to change after Initial Access. ``Beam Tracking'' methods are commonly used to correct for smaller angular changes and maintain the viability of active link(s) \cite{li2017analog,hashemi2018efficient}.
Nonetheless, due to the narrow beams at both TX and RX, established communication links are prone to blockage (by objects in the communication environment, and even the users themselves).
Hence, the initial link establishment stage might need to be repeated multiple times during every communication session.
This results in high overhead to establish coherent beams during the course of the session, if the initial access process is inefficient.
This paper focuses on the Initial Access problem.

\textbf{Compressed Sensing (CS):}
In CS theory, the main objective is to recover an unknown \textit{sparse vector} $\boldsymbol{q^a}$ using a small number (compared to the sparse vector dimensions) of linear measurements.
Measurements in CS, denoted by $\boldsymbol{y},$ are modeled as $\boldsymbol{y} = \boldsymbol{B} \boldsymbol{q^a}$, where $\boldsymbol{B}$ is the sensing matrix.
Hence, $\boldsymbol{B}$ is a linear transformation that amounts to encoding $\boldsymbol{q^a}$. Sparse recovery algorithms, on the other hand, amount to decoding $\boldsymbol{y}$.
To obtain ``good'' measurements (which best preserve the information contained in the channel matrix), the sensing matrix need to be stochastically optimized based on criteria like the $\spark{\hspace{-2pt}(\boldsymbol{B})}$ (i.e., minimum number of linearly dependent columns), the mutual coherence and the Restricted Isometry Property.

Since mmWave channel matrices are sparse, and since channel measurements are linear operations, CS became a dominant approach for tackling mmWave channel estimation problems. The main caveat here is that the standard CS problem is that of a sparse vector recovery, while mmWave channel estimation is a sparse matrix recovery.
This distinction poses some challenges in tackling mmWave channel estimation under the umbrella of CS.
To formulate MIMO channel estimation as a CS problem, a vectorization step is carried out (i.e., columns of matrices are stacked on top of each other to form one long vector).
Nonetheless, unlike standard CS problems in which elements of the sensing matrices are directly chosen and optimized, the mmWave sensing matrix is a function of the transmit precoding and receive combining vectors. 
This adds an extra layer of complexity which is often ignored under the premise that since CS often requires random sensing matrices, then random beamforming is an obvious necessity.
However, it is not immediately clear how a specific choice of precoders and combiners would affect the structure of $\boldsymbol{B}$, and therefore, the performance of sparse recovery.
Extending the design principles of sensing matrices from core CS theory to mmWave channel estimation is thus not straightforward and remains an open area of research.

Existing research on CS-based mmWave channel estimation relies on \textit{\textbf{random}} arbitrary choices of precoding and combining vectors, e.g. uniformly distributed phase shifts \cite{Alkhateeb_2015_HowManyMeasurements, lu2019comparison}.
When this solution is incorporated in mmWave channel estimation, it translates into designing antenna beam patterns of highly irregular shapes (see Fig. \ref{fig:rx_CS}). Such beam patterns are sensitive to variations of received signal power, thermal noise and resolution of ADCs and phase-shifters.
Our proposed source-coding-based solution overcomes these limitations by imposing better, well-structured antenna patterns, where, in each measurement, a specific angular direction is either included (with constant beamforming gain) or is excluded.
This provides better resilience to i) the presence of sidelobes, ii) variations in received signal power along any available path(s), iii) channel noise and iv) quantization error of ADCs and phase shifters.
Furthermore, the deterministic nature of our source-coding-based measurements allows us to provide theoretical guarantees for channel recovery at a precise number of measurements.
The source coding analogy also allows us to draw theoretical tight lower bounds on the number of measurements.

On the contrary, the number of required measurements in CS is commonly characterized as an order of magnitude.
For instance, several state-of-the-art sparse recovery algorithms require $O(L\log(\frac{n}{L}))$ measurements, where $n$ is the number of dimensions of the sparse vector and $L {\ll} n$ is its sparsity level \cite{CS4Wireless_TipsAndTricks, alkhateeb2014channel}.
This, however, is just a scaling law, which by definition, works in the asymptotic regime and is missing the constant scaling coefficient.
Compare this to our solution, which accurately specifies the required number of measurements (based on $n$ and $L$).

Developing efficient sparse recovery algorithms for CS-based mmWave channel estimation is a rich area of research.
Various algorithms have different computational complexities, recovery performance, favorable range of signal to noise ratio (SNR), etc.
A comparison between several classes of sparse recovery algorithms is provided in \cite{lu2019comparison}. These include convex relaxation (e.g. $\ell_1$ norm minimization), greedy iteration (e.g. Orthogonal Matching Pursuit (OMP)) and Bayesian Inference.
Other algorithms also include Approximate Message Passing (AMP) \cite{wen2016bayes} and its variants \cite{mo2018channel},
as well as machine learning based sparse recovery \cite{ma2019deep}.

\textbf{Machine Learning:}
Deep learning is very powerful in extracting patterns from large amounts of data.
It has been widely used in problems of computer vision, speech recognition and natural language processing.
Recently, it has also been applied to problems in communications \cite{wang2017deep}, including, but not limited to channel estimation \cite{he2018deep, huang2018deep, alkhateeb2018deep}.
For instance, in \cite{alkhateeb2018deep} the beamforming vectors at the TX and RX are \textit{``learned''} based on uplink pilot signals simultaneously received at multiple base stations.
The base stations share their received information on a cloud, on which data processing is performed.
This idea, however, is critically dependent on a dense deployment of base stations.
In \cite{he2018deep, huang2018deep}, deep learning is leveraged to ease the burden of heavy computations that would otherwise be required for measurement processing.

\textbf{Coding:}
Exploiting source codes for mmWave channel estimation has not been studied before.
Our earlier work in \cite{shabaraBeamToN} drew an analogy between \textit{path discovery} in mmWave channels and \textit{error discovery} in \textit{Linear Block Channel Codes} (LBC).
There exists a \textit{duality} between the error discovery problem of channel codes and linear source compression.
That is, we can use LBCs as \textit{linear} source compression codes, as well.
Nonetheless, the channel coding analogy did not naturally lend itself to characterizing the lower bound on the required number of measurements. This paper differs from \cite{shabaraBeamToN} in the following:
\begin{itemize}[noitemsep,topsep=0pt]
\item Channel measurements are envisaged as compressed versions of the channel, which are obtained based on lossless, fixed-rate, linear source codes.
\item The lower bound on the achievable number of measurements is accurately characterized using the entropy of the direction of the strong reflectors (in a stochastic spatial model). This not only provides a precise metric to quantify the efficiency of a used code, but it also provides a benchmark for evaluating other measurement schemes. 
\item A DNN-based measurement decoding is proposed and evaluated against the more computationally complex ``search'' method of \cite{shabaraBeamToN}.
\item A comparison to compressed sensing based mmWave channel estimation is provided, demonstrating the superiority of our proposed approach.
\end{itemize}

\textbf{Hashed Beams:} An idea of direction inclusion/exclusion used to generate antenna patterns was adopted in \cite{hassanieh2018fast}.
Specifically, every measurement combines signals coming from a \textbf{\textit{randomly}} chosen set of angular directions. This describes the encoding part and it resembles a \textit{random} binary code. For decoding, a threshold-based decision determines whether a strong propagation path exists (if a path exists, it lies at one of the directions included in this measurement).
The direction which was most frequently included in the measurements that revealed a strong path is declared as the angular direction of the strongest channel path.
This method discovers one path, and requires $O(L\log(n))$ measurements.
In our proposed approach, however, the angular directions whose respective beams are overlapped are \textbf{\textit{precisely determined}} using a carefully chosen code.
We also use an elaborate decoding method that is capable of discovering multiple paths.
It also guarantees a lower number of measurements since a randomly chosen code is not expected to outperform a carefully designed one.

\section{Motivating Example}
\begin{figure}[t]
\centering
\includegraphics[width=0.4\linewidth]{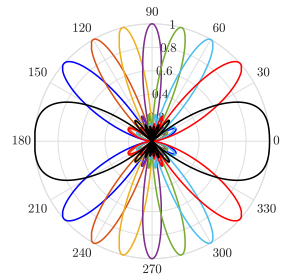}
\caption{\small Antenna sectors corresponding to angular directions $\{d_i\}_{i = 1}^{8}$.}
\label{fig:antennaSectors} 
\end{figure}

\begin{table}[t]
\caption{\small Channel measurements $\boldsymbol{y^s}$ corresponding to all $\boldsymbol{q^a} {\in} \mathcal{Q}^a$}
\label{table:ChannelToMeasurmentMap}
\centering
 \begin{tabular}{||c | c | c ||} 
 \hline
\multicolumn{2}{||c|}{Angular Channel $\boldsymbol{q^a}^T$} &
\multicolumn{1}{ c||}{Channel measurement ${\boldsymbol{y}^{\boldsymbol{s}}}^T$} \\
 \hline\hline
$d_0$ &$[0 \: 0 \: 0 \: 0 \: 0 \: 0 \: 0 \: 0]$ & $[0 \: 0 \: 0 \: 0]$\\ 
 \hline
$d_1$ & $[1 \: 0 \: 0 \: 0 \: 0 \: 0 \: 0 \: 0]$ & $[1 \: 0 \: 0 \: 0]$\\
 \hline
$d_2$ & $[0 \: 1 \: 0 \: 0 \: 0 \: 0 \: 0 \: 0]$ & $[0 \: 1 \: 0 \: 0]$\\
 \hline
$d_3$ & $[0 \: 0 \: 1 \: 0 \: 0 \: 0 \: 0 \: 0]$ & $[0 \: 0 \: 1 \: 0]$\\
 \hline
\whencolumns
{
\vdots & \vdots & \vdots \\
}
{
$d_4$ & $[0 \: 0 \: 0 \: 1 \: 0 \: 0 \: 0 \: 0]$ & $[0 \: 0 \: 0 \: 1]$\\
 \hline
$d_5$ & $[0 \: 0 \: 0 \: 0 \: 1 \: 0 \: 0 \: 0]$ & $[1 \: 1 \: 0 \: 0]$\\
 \hline
$d_6$ & $[0 \: 0 \: 0 \: 0 \: 0 \: 1 \: 0 \: 0]$ & $[0 \: 1 \: 1 \: 0]$\\
 \hline
$d_7$ & $[0 \: 0 \: 0 \: 0 \: 0 \: 0 \: 1 \: 0]$ & $[0 \: 0 \: 1 \: 1]$\\}
 \hline
$d_8$ & $[0 \: 0 \: 0 \: 0 \: 0 \: 0 \: 0 \: 1]$ & $[1 \: 1 \: 0 \: 1]$\\ 
\hline
\end{tabular}
\end{table}

Consider an RX equipped with an antenna array which can form $8$ distinct beams. These 8 beams divide the angular space into resolvable directions, i.e., $d_1, d_2, \dots, d_8$, as shown in Fig. \ref{fig:antennaSectors}.
The RX needs to establish a Line of Sight (LoS) communication link with TX.
This requires some sort of ``searching'' over the angular space at both TX and RX.
For ease of illustration, let us reduce the link establishment problem to that of ``Angle of Arrival (AoA) discovery'' at RX by assuming that TX transmits its signals omnidirectionally.
Let the path gain of LoS be denoted by $\alpha$, which can take arbitrary values. For simplicity of notation, let us assume $\alpha = 1$.

Our \textbf{\textit{objective}} is: Find the specific direction $d_{i^*}$ which contains the LoS path to TX using the \textbf{\textit{least possible number of measurements}}.
To do so, we envisage the measurement process as \textbf{\textit{lossless}}, \textbf{\textit{fixed-rate}} \textbf{\textit{channel compression}}.
This enables \textit{harnessing the power of source compression codes to minimize the number of measurements}.
It also enables deriving lower bounds on the number of measurements, using which we can accurately find the LoS (or conclude it is blocked).
We propose a measurement approach which has a \textbf{\textit{predetermined}} measurement sequence that (1) does not require feedback and (2) is capable of finding the LoS path, no matter in which $d_i$ it exists. Therefore, a constant number of measurements, $m$, is needed for all $d_i$.

\begin{figure*}[t]
\centering
\includegraphics[width=1\linewidth]{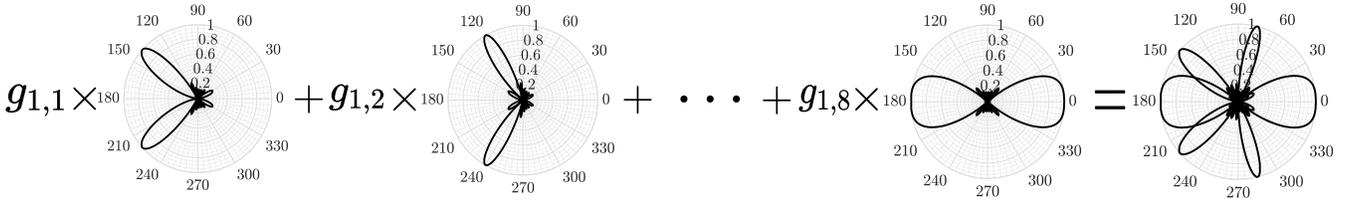}
\caption{\small Generating the beam pattern required for the $1^{st}$ measurement ($y^s_1$)}
\label{fig:beamforming_demo}
\end{figure*}

The key idea of LoS discovery using non-feedback linear source coding is to:
1) Construct a \textit{binary} codebook that represents the angular channel,
2) Find a proper fixed-rate linear source code that losslessly compresses all codewords in that codebook, and
3) Use this code to design the measurements.
These steps can be elucidated as follows:
\textbf{i)} Constructing the codebook is as easy as finding all possible \textit{binary} vectors that represent the LoS position.
Since $\alpha {=} 1$, this codebook is exactly the set of all possible channel vectors.
Let the channel between TX and RX be denoted by $\boldsymbol{q^a}$, and let $\mathcal{Q}^a$ be the set of all possible channels.
The channel $\boldsymbol{q^a} {\in} \mathcal{Q}^a$ has $8$ components; each one represents the path gain corresponding to a unique angular sector as shown in Fig. \ref{fig:antennaSectors}.
Table \ref{table:ChannelToMeasurmentMap} shows all possible $\boldsymbol{q^a}$ in our setup (for arbitrary gain values, simply replace the `1's in Table \ref{table:ChannelToMeasurmentMap} with $\alpha$).
\textbf{ii)} Choose the linear source code, denoted by its generator matrix $\boldsymbol{G}$ as:
\begin{equation}
\renewcommand{\arraystretch}{0.7}
\label{eqn:codeChoice}
\small
  \boldsymbol{G} =
  \begingroup
  \setlength\arraycolsep{6pt}
  \begin{pmatrix}
  1 & 0 & 0 & 0 & 1 & 0 & 0 & 1\\
  0 & 1 & 0 & 0 & 1 & 1 & 0 & 1\\
  0 & 0 & 1 & 0 & 0 & 1 & 1 & 0\\
  0 & 0 & 0 & 1 & 0 & 0 & 1 & 1\\
  \end{pmatrix}
  \endgroup
\end{equation}

To compress $\boldsymbol{q^a}$, we simply need to find the matrix multiplication $\boldsymbol{y^s} {=} \boldsymbol{G} \boldsymbol{q^a}$ (see Table \ref{table:ChannelToMeasurmentMap}).
\textbf{iii)} Design the measurements such that $\boldsymbol{y^s}$ is imitated by the measurement results.
\textbf{\textit{This is done by beamforming at RX}}.
Notice that the $i^{\text{th}}$ measurement, i.e., $y^s_i$ (the $i^{\text{th}}$ component of $\boldsymbol{y^s}$) is the multiplication of the $i^{\text{th}}$ row of $\boldsymbol{G}$ by $\boldsymbol{q^a}$.
Mathematically, this is just adding all elements $q^a_j$ of $\boldsymbol{q^a}$ which corresponds to $g_{i,j} {=} 1$, ($g_{i,j}$ is the element at row $i$, and column $j$ in $\boldsymbol{G}$). That is
\begin{equation}
\label{eqn:beamformingObjective}
\small y^s_i = \sum_{\substack{j = 1}}^8 q^a_j \times g_{i,j} = \sum_{j:\;g_{i,j}=1} q^a_j.
\end{equation}
Hence, measurement $i$ should only contain the directions $d_j$ whose corresponding $g_{i,j}$ equals $1$, and exclude the rest (Notice that we can map the $i^{\text{th}}$ row of $\boldsymbol{G}$ to the $i^{\text{th}}$ measurement and the $j^{\text{th}}$ column to the $j^{\text{th}}$ sector (direction $d_j$)).
Essentially, this means that in each of the measurements $y^s_i$, we combine the signals received at a specific set of AoA directions
This can be realized by carefully shaping the antenna pattern using beamforming.
Fig. \ref{fig:beamforming_demo} highlights this process for the $1^{\text{st}}$ measurement in which only the direction $d_1, d_5$ and $d_8$ are included.
\textit{The measurement results $\forall d_i$ are shown in Table \ref{table:ChannelToMeasurmentMap}}.
Note that the number of required measurements is $4$ for all $d_i$.

\textbf{Lower Bound:}
A fundamental question that arises here is: \textit{Can we find a better fixed-rate, lossless source code (other than the one given in Eq. (\ref{eqn:codeChoice})) that would produce fewer measurements, and hence increase the efficiency of the measurement process}?
To answer this question, we need to find \textbf{the minimum expected number of measurements} required to discover the channel using our proposed source coding solution.
This number is identical to the minimum average code length (over all fixed-rate, lossless codes).
The minimum average code length is well-known to be lower bounded by the \textbf{\textit{Shannon Entropy}}; denoted by $H_2$ and defined as
\begin{equation}
\label{eqn:entropy}
\small H_2\left(\boldsymbol{q^a}\right) 
= \small \sum_{\boldsymbol{q^a} \in \mathcal{Q}^a} \mathbb{P} \left( \boldsymbol{q^a} \right) \log_2\left( \frac{1}{\mathbb{P} \left( \boldsymbol{q^a} \right)} \right)
\end{equation}
Calculating $H_2\left(\boldsymbol{q^a}\right)$ requires knowledge of the probability distribution $\mathbb{P} \left( \boldsymbol{q^a} \right)$.
Fixed-rate codes, however, do not account for the frequency of $\boldsymbol{q^a}$ (hence, the mapping to equal-length codes).
By limiting the space of codes to be over fixed-rate codes, we can improve the bound to be
\begin{equation}
\label{eq:LB}
\ceil*{\log_2\left(\vert\mathcal{Q}^a\vert \right)} \geq H_2\left(\boldsymbol{q^a}\right)
\end{equation}
where $\vert\mathcal{Q}^a\vert = 9$ (recall that there exists $9$ possible scenarios for $\boldsymbol{q^a}$ as shown in Table \ref{table:ChannelToMeasurmentMap}).
This tighter bound is obtained by assuming a \textit{Uniform} distribution,  which is the entropy maximizing distribution, over the channel space $\mathcal{Q}^a$.
Eq. (\ref{eq:LB}) reveals that our chosen code achieves the lower bound of $4$ measurements.
We provide a formal discussion on the lower bound in Section \ref{sec:LB}.

\begin{remark}
This Motivating Example only dealt with a simplified channel model, with only one channel path and a fixed path gain of $\alpha = 1$.
However, in the rest of the paper, we will consider generalized channel models with possibly several paths of arbitrary path gain values, i.e., $\alpha {\in} \mathbb{C}$.
\end{remark}

\section{System Model}
\label{sec:systemModel}
\begin{figure}[t]
\centering
\includegraphics[width=1\linewidth]{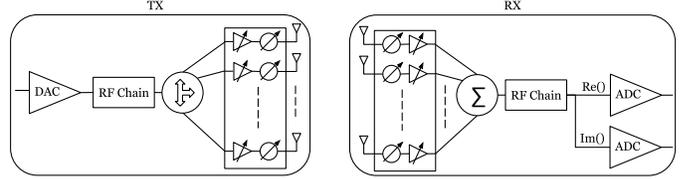}
\caption{\footnotesize Transceiver architecture: At TX, an $n_t$-way power splitter divides the transmit signal which is then passed through variable-gain amplifiers and phase-shifters. A single DAC is required since TX sends real valued signals. At RX, the acquired signal is passed through a similar network of power amplifiers and phase-shifters before being combined and fed to a single RF chain. Two ADCs are required to obtain I/Q components of received signals.}
\label{fig:architecture} 
\end{figure}
We consider point-to-point mmWave channels with $n_t$ and $n_r$ antennas at TX and RX, respectively.
Antennas at TX and RX form Uniform Linear Arrays (ULA). Generalization to Uniform Planar Arrays is straightforward but not considered in this paper for simplicity. Every antenna element is connected to a phase-shifter and a low-power variable-gain amplifier (VGA)\footnote{The use of VGAs in analog transceivers is common in practice. For instance, in IEEE 802.11ad \cite{IEEE80211ad} both phase and amplitude components are used to specify antenna weights, and commercial devices like Wilocity Wil6200 offer this capability. VGAs are also used along with phase-shifters in practice to help compensate for their phase-dependent insertion loss \cite{li201360}.}.
On the TX side, a single RF chain feeds its ULA through an $n_t$-way power splitter, while on the RX side, the outputs of the ULA, after being processed by amplifiers and phase-shifters, are then linearly combined using an adder and fed through to a single RF chain with in-phase (I) and quadrature (Q) channels.
Two mid-tread ADCs with $2^b{+}1$ levels are used to quantize the I and Q components of the received signal. The term $b$ loosely denotes the number of bits that describe the ADC resolution.
Fig. \ref{fig:architecture} depicts the transceiver architecture.

We assume single-tap channels where all channel paths have just one significant tap.
We also adopt a channel clustering model where paths between TX and RX form clusters in the angular domain \cite{rangan2014millimeter, rappaport2013millimeter}.
Let $L$ denote the number of available channel clusters.
Due to the sparse nature of mmWave channels, only a limited number of clusters exist\footnote{Prior knowledge about the number of clusters can be obtained from statistical channel information, which in turn are obtained from channel measurement campaigns. For instance, measurements carried out in New York City revealed that an average number of $2$ or $3$ clusters exists in mmWave channels at $28$ and $73$ GHz \cite{akdeniz2014millimeter}.}, where $L \ll n_r,n_t$.
Since distinct paths within each cluster cannot typically be resolved, we assume that each cluster contains only one path.
Each channel path (e.g., $l^{th}$ path) is attributed with an AoD $\theta_l$, an AoA $\phi_l$ and a path gain $\alpha_l$.
Let $\alpha_l^b \in \mathbb{C}$ denote the baseband path gain such that
\begin{equation}
\label{eqn:basebandPathGain}
\alpha_l^b = \alpha_l \sqrt{n_r n_t} e^{-j \frac{2 \pi \rho_l}{\lambda_c}},
\end{equation}
where $\rho_l$ is the path length and $\lambda_c$ is the carrier wavelength.
We define the \textit{directional cosines} of the AoD and AoA of the $l^{th}$ path as $\Omega_{tl} {\triangleq} \cos \left( \theta_l \right)$ and $\Omega_{rl} {\triangleq} \cos \left( \phi_l \right)$, respectively.
The transmit and receive spatial signatures at an arbitrary directional cosine $\Omega$ is denoted by $\boldsymbol{e_t}\left( \Omega \right)$ and $\boldsymbol{e_r}\left( \Omega \right)$, receptively.
We define $\boldsymbol{e_t}\left( \Omega \right)$ (and similarly $\boldsymbol{e_r}\left( \Omega \right)$) as:
\begin{equation}
\small
\boldsymbol{e_t} \left( \Omega \right) {=}
\frac{1}{\sqrt{n_t}}
\left(
1 ,
e^{ - j 2 \pi   \Delta_t \Omega} ,
e^{ - j 2 \pi 2 \Delta_t \Omega} ,
\hdots ,
e^{ - j 2 \pi (n_t-1) \Delta_t \Omega}
\right)^T 
\end{equation}\normalsize
where $\Delta_t$ and $\Delta_r$ are the antenna separations at TX and RX ULAs, normalized by $\lambda_c$.

Let $\boldsymbol{Q} \in \mathbb{C}^{n_r {\times} n_t}$ denote the channel matrix such that
\begin{equation}
\small \boldsymbol{Q} = \sum_{l=1}^{L} \alpha_l^b \boldsymbol{e_r}\left( \Omega_{rl} \right) \boldsymbol{e}_{\boldsymbol{t}}^H\left( \Omega_{tl} \right).
\end{equation}
The corresponding angular channel of $\boldsymbol{Q}$, whose rows and columns divide the channel into resolvable RX and TX angular bins, respectively, is denoted by $\boldsymbol{Q^a}$ and can be obtained as:
\begin{equation}
\label{eqn:Qa_from_Q}
\boldsymbol{Q^a} = \boldsymbol{U}_{\boldsymbol{r}}^H \boldsymbol{Q} \boldsymbol{U_t}.
\end{equation}
If $n_t$ or $n_r$ equals $1$, $\boldsymbol{Q}$ and $\boldsymbol{Q^a}$ are reduced to vectors which we denote by $\boldsymbol{q}$ and $\boldsymbol{q^a}$, respectively.
The matrices $\boldsymbol{U_t}$ and $\boldsymbol{U_r}$ are the transmit and receive \textit{unitary Discrete Fourier Transform} (\textit{DFT}) matrices whose columns form an orthonormal basis for the transmit and receive signal spaces $\mathbb{C}^{n_t}$ and $\mathbb{C}^{n_r}$, respectively.
The definition of $\boldsymbol{U_t}$ (likewise for $\boldsymbol{U_r}$) is given as \cite[Chapter 7.3.4]{tse2005fundamentals}
\begin{equation}
\boldsymbol{U_t} \triangleq
\begin{pmatrix}
\boldsymbol{e_t} \left( 0 \right) &
\boldsymbol{e_t} \left( \frac{1}{L_t} \right) &
\hdots &
\boldsymbol{e_t} \left( \frac{n_t-1}{L_t} \right)
\end{pmatrix},
\end{equation}
where $L_t {=} n_t \Delta_t$ ($L_r {=} n_r \Delta_r$) denote the length of the TX (RX) antenna array, normalized by $\lambda_c$.

Similar to \cite{alkhateeb2014channel, ChEstScniter14, RACE_2017}, we assume perfect sparsity where channel paths lie along AoD and AoA directions defined in $\boldsymbol{U_t}$ and $\boldsymbol{U_r}$. Hence, each path only contributes to a single component of $\boldsymbol{Q^a}$.
Thus, only $L$ (possibly less) non-zero components in $\boldsymbol{Q^a}$ exists.
The baseband model is
\begin{equation}
\boldsymbol{u_b} = \boldsymbol{Q}\boldsymbol{x} + \boldsymbol{n}
\end{equation}
where $\boldsymbol{u_b}$ is the received vector at RX front-end while $\boldsymbol{n} {\sim} \mathcal{CN}\left( \boldsymbol{0} , N_0 \boldsymbol{I}_{n_r} \right)$ is an i.i.d. complex Gaussian noise vector.
TX sends pilot symbols $s$, with power $P$, which are processed using precoders $\boldsymbol{f}_j {\in} \mathbb{C}^{n_t}$ to obtain the transmit vectors $\boldsymbol{x} {=} \boldsymbol{f}_j s$.
Hence, the transmit SNR is
\begin{equation}
\label{eqn:SNR}
\small \text{SNR} \triangleq \frac{P}{N_0} \times \mu,
\end{equation}
where $\mu$ is the average path loss (which depends on the carrier frequency, atmospheric conditions, average distance between TX and RX). Note that $\text{SNR}$ and $\mu$ are not path dependent.
The rx-combining vectors $\boldsymbol{w}_i {\in} \mathbb{C}^{n_r}$ are used to obtain the received symbols $u_{i,j}$ such that
\begin{equation}
u_{i,j} = \boldsymbol{w}_i^H \boldsymbol{Q} \boldsymbol{f}_j s  + \boldsymbol{w}_i^H \boldsymbol{n}
\label{eqn:error-free_symbols}
\end{equation}
where $i \in \{1,\dots,m_r\}, j \in \{1,\dots,m_t \}$.
Finally, a quantized version $u^s_{i,j}$ of $u_{i,j}$ is obtained such that
\begin{equation}
\label{eqn:noisyMeasurements}
u^s_{i,j} = \left[ \boldsymbol{w}_i^H \boldsymbol{Q} \boldsymbol{f}_j s + \boldsymbol{w}_i^H \boldsymbol{n} \right]_+,
\end{equation}
where $[\cdot]_+$ represents the quntization function.
The noise component, normalized by $\norm{\boldsymbol{w}_i}$ has a complex Gaussian distribution, i.e., $\frac{\boldsymbol{w}_i^H \boldsymbol{n}}{\norm{\boldsymbol{w}_i}} \sim \mathcal{CN}\left( 0 , N_0 \right)$.

Let $y^s_{i,j} = \boldsymbol{w}_i^H \boldsymbol{Q} \boldsymbol{f}_j s$ denote the error-free measured symbols and let $z_{i,j} = u^s_{i,j} - y^s_{i,j}$ denote the \textit{measurement error} which includes both \textbf{channel noise} and \textbf{quantization error}.

\subsection{Problem Formulation}
The problem we need to solve is to minimize the number of measurements $m = m_t \times m_r$ such that $\boldsymbol{Q}$ can be reliably reconstructed.
This problem can be mathematically stated as:
\begin{mini!}|l|[2]
{\substack{\boldsymbol{w_i}, \boldsymbol{f_j}, \mathcal{D} }}
         {m_t \times m_r} {}{P1:}
\addConstraint{y^s_{i,j} }{{=} \boldsymbol{w}_i^H \boldsymbol{Q} \boldsymbol{f}_j}
\addConstraint{\mathcal{D}(\left\lbrace y^s_{i,j}}\right\rbrace){ {=} \boldsymbol{Q}^a} \label{eqn:decodingFn}
\end{mini!}
Note that $y^s_{i,j}$ exists $\forall i,j \in \{1,\dots,m_r\} \times \{1,\dots,m_t\}$. That is, measurements are taken using all combinations of $\boldsymbol{f}_j$ and $\boldsymbol{w}_i$. We also use $s = 1$.
The design variables are the tx-precoders $\boldsymbol{f_j}$, the rx-combiners $\boldsymbol{w_i}$ and the decoding function $\mathcal{D}$.
We do not explicitly consider the impact of errors in this formulation but its effect will be studied in Section \ref{sec:errorAnalysis}.
Note also that due to the use of VGAs at each antenna element, the constant modulus constraint on $\boldsymbol{f_j}$ and $\boldsymbol{w_i}$, that is often incorporated in analog beamforming designs, is not needed.

\section{Source-Coding-Based Measurements}
\label{sec:lowerBound}
In this section, we formally introduce \textit{mm-wave beam discovery} as a source coding problem.
We initially focus on channels with single-transmit, multiple-receive antennas.
Specifically, we provide the conditions under which a chosen fixed-rate source code can be used to uniquely \textit{``encode''} channel vectors in $\mathbb{C}^{n_r}$ into measurement vectors of fewer components.
This setting is \textbf{\textit{identical to that of multiple-transmit, single-receive antennas.}}
In Section \ref{sec:DNN}, we show how to use DNNs to \textit{``decode''} the measurements and 
obtain an estimate for the observed channel.
Then, in Section \ref{sec:multiTxRx}, we consider general channels with multiple TX and RX antennas.
Now, let us start with the following discussion on source codes. 

\subsection{Source Codes}
Let $C$ be a binary \textit{linear} source code with encoding and decoding functions denoted by $\mathcal{E}_{\mathbb{F}_2}$ and $\mathcal{D}_{\mathbb{F}_2}$, respectively.
We refer to $C$ as the encoding-decoding function pair $\left(\mathcal{E}_{\mathbb{F}_2},\mathcal{D}_{\mathbb{F}_2}\right)$.
The subscript $\mathbb{F}_2$ denotes the finite field of two elements $\mathsf{0}_{\mathbb{F}_2}$ and $\mathsf{1}_{\mathbb{F}_2}$ (also referred to as $GF(2)$) over which the code $C$ is defined.
Later on, we will drop the subscripts to simplify notation as long as they can be inferred from the context.

\begin{definition}[Linear Source Code]
A source code $C$ whose encoding function $\mathcal{E}_{\mathbb{F}_2}$ is a linear function of the source sequences is called a \textbf{\textit{linear source code}}.
\end{definition}

Let $\boldsymbol{s}$ be a source sequence of length $n$ where $\boldsymbol{s} \in \mathcal{S} \subseteq \left\lbrace \mathsf{0}_{\mathbb{F}_2},\mathsf{1}_{\mathbb{F}_2} \right\rbrace^n$, and let $\boldsymbol{c_s} \in \mathcal{I}_{\mathcal{S}} \subseteq \left\lbrace \mathsf{0}_{\mathbb{F}_2},\mathsf{1}_{\mathbb{F}_2} \right\rbrace^m$ be its associated binary representation under $C$ where $\mathcal{I}_{\mathcal{S}}$ is the image of $\mathcal{S}$ under $\mathcal{E}_{\mathbb{F}_2}$.
Thus, using a linear source code $C$, we can find the representation of $\boldsymbol{s}$ under $C$ using
\begin{equation}
\boldsymbol{c_s} = \boldsymbol{G} \boldsymbol{s},
\end{equation}
where $\boldsymbol{G} \in \left\lbrace \mathsf{0}_{\mathbb{F}_2},\mathsf{1}_{\mathbb{F}_2} \right\rbrace ^{m \times n}$ is called the \textbf{\textit{generator matrix}}. Note that linearity guarantees fixed-rate since the code length is a constant value (equals the number of rows of $\boldsymbol{G}$).

The decoding function $\mathcal{D}_{\mathbb{F}_2}$, maps sequences $\boldsymbol{c_s}$ to a corresponding source sequence $\hat{\boldsymbol{s}} \in \hat{\mathcal{S}} \subseteq \left\lbrace \mathsf{0}_{\mathbb{F}_2},\mathsf{1}_{\mathbb{F}_2} \right\rbrace ^n$.
Suppose that $\mathcal{S}$ is the set of all sequences such that if $\boldsymbol{s}_1, \boldsymbol{s}_2 \in \mathcal{S}$, we have that $ \boldsymbol{s}_1 \neq \boldsymbol{s}_2 \xLongleftrightarrow{\text{iff}} \boldsymbol{c_s}_1 \neq \boldsymbol{c_s}_2$.
In other words, $\mathcal{E}_{\mathbb{F}_2}: \mathcal{S} \rightarrow \mathcal{I}_{\mathcal{S}}$ is injective (one-to-one).
Consequently, if we define the function $\mathcal{D}_{\mathbb{F}_2}$ over $\mathcal{I}_{\mathcal{S}}$ as the inverse function of $\mathcal{E}_{\mathbb{F}_2}$, i.e., $\mathcal{E}_{\mathbb{F}_2}^{-1} \triangleq \mathcal{D}_{\mathbb{F}_2} : \mathcal{I}_{\mathcal{S}} {\rightarrow} \mathcal{S}$ , then we have that
$\hat{\boldsymbol{s}} = \mathcal{D}_{\mathbb{F}_2}\left( \mathcal{E}_{\mathbb{F}_2} \left( \boldsymbol{s} \right) \right) = \boldsymbol{s}, \; \; \forall \boldsymbol{s} \in \mathcal{S}$.

\subsection{MmWave Beam Discovery}
\label{sec:LB_mmwave}
Let $\boldsymbol{q^a} \in \mathbb{C}^{n_r}$ denote the angular \textbf{\textit{channel vector}} between TX and RX. Define $\boldsymbol{q}^a_s \in \left\lbrace 0,1 \right\rbrace^{n_r}$ to be the \textbf{\textit{support vector}} associated with $\boldsymbol{q}^a$ such that
$\boldsymbol{q}^a_s = \begin{pmatrix}
q^a_{s_1} &
q^a_{s_2} &
\dots &
q^a_{s_{n_r}}
\end{pmatrix}^T$
where $q^a_{s_i} = 1$ if  $q^a_i \neq 0$ and $q^a_{s_i} = 0$ otherwise.
More generally, a support vector can be defined as:
\begin{definition} [Support vector]
The support vector $\boldsymbol{v_s}$ associated with an arbitrary $n-$dimensional vector $\boldsymbol{v} \in \mathbb{C}^n$ is a binary vector of the same size that identifies the non-zero components of $\boldsymbol{v}$ and whose components, $v_{s i}$, are defined as $v_{s i} = 1$ if $v_i \neq 0 $ and $v_{s i} = 0$ if $v_i = 0 $.
\end{definition}

We further define the set of non-zero indexes $\mathcal{X}_{\boldsymbol{v}}$ of an arbitrary vector $\boldsymbol{v}$ as follows:
\begin{definition}[Set of Non-Zero Indexes $\mathcal{X}_{\boldsymbol{v}}$]\label{def:setsOfIndexes}
For any arbitrary $n-$dimensional vector $\boldsymbol{v}$, we define $\mathcal{X}_{\boldsymbol{v}}$ as the set of indexes of its non-zero components, i.e.,
$\mathcal{X}_{\boldsymbol{v}} = \left\lbrace i \vert v_i \neq 0 \; , 0 {\leq} i {\leq} n{-}1 \right\rbrace$.
\end{definition}
Hence, if $\boldsymbol{v}_s$ is the support vector corresponding to $\boldsymbol{v}$, then we have that $\mathcal{X}_{\boldsymbol{v}} =\mathcal{X}_{\boldsymbol{v}_s}$, since $v_i = 0 {\iff} {v}_{s_i} = 0$.
Now, let $\mathcal{Q}^a$ be the set containing all possible channel vectors $\boldsymbol{q}^a$.
Also let $\mathcal{Q}^a_s$ be the set of all support vectors $\boldsymbol{q}^a_s$ such that their corresponding channels $\boldsymbol{q}^a {\in} \mathcal{Q}^a$.
An interesting behavior we have for these sets is as follows: If we have a channel $\boldsymbol{q}^{\boldsymbol{a}}_1$ whose support vector $\boldsymbol{q}^a_{s_1} \in \mathcal{Q}^a_s$, then removing any non-zero component(s) from $\boldsymbol{q}^{\boldsymbol{a}}_1$ (due to blockage for example) would still yield a valid channel $\boldsymbol{q}^{\boldsymbol{a}}_2 \in \mathcal{Q}^a$, whose support vectors $\boldsymbol{q}^{\boldsymbol{a}}_{s_2}$ also belongs to $\mathcal{Q}^a_s$.
We call this the \textbf{\textit{inclusion property}}.

\begin{definition}\label{remark:inclusion}
[Inclusion Properties of $\mathcal{Q}^a_s$]
\begin{enumerate}[(i)]
\item Let $\boldsymbol{q}^a_{s_1}, \boldsymbol{q}^a_{s_2} \in \left\lbrace0,1\right\rbrace^{n_r}$ such that
$\mathcal{X}_{\boldsymbol{q}^a_{s_2}} \subseteq \mathcal{X}_{\boldsymbol{q}^a_{s_1}}$.
If $\boldsymbol{q}^a_{s_1} \in \mathcal{Q}^a_s$, then $\boldsymbol{q}^a_{s_2} \in \mathcal{Q}^a_s$. \label{prop1}
\item $\boldsymbol{0} \in \mathcal{Q}^a_s$. In fact, this is a consequence of property \ref{prop1} above since for any $\boldsymbol{q}^a_s \in \mathcal{Q}^a_s$, we have that $\mathcal{X}_{\boldsymbol{0}} {=} \varnothing {\subseteq} \mathcal{X}_{\boldsymbol{q}^a_s}$.
\end{enumerate}
\end{definition}

Now, we are ready to present the theorem that establishes the conditions that need to be satisfied by a linear source code so that each possible channel $\boldsymbol{q^a}$ would result in a \textbf{\textit{unique}} measurement vector $\boldsymbol{y^s}$. Impairments under noise are not addressed in this theorem.
\begin{theorem}\label{thm:uniquenessOfMeasurements}
Consider a \textbf{\textit{binary}} linear source code $C$ whose encoding function $\mathcal{E}$ (defined by the binary generator matrix $\boldsymbol{G}$) is an injective function defined over $\mathcal{Q}^a_s \in \left\lbrace 0,1 \right\rbrace ^{n_r}$.
If we consider $\boldsymbol{G}$ to be defined over the complex field, then for all channel vectors $\boldsymbol{q}^a_1, \boldsymbol{q}^a_2 \in \mathcal{Q}^a \subseteq \mathbb{C}^{n_r}$ we have $\boldsymbol{q}^a_1 \neq  \boldsymbol{q}^a_2 \textit{ if and only if } \boldsymbol{G} \boldsymbol{q}^a_1 = \boldsymbol{y}_1^{\boldsymbol{s}} \neq \boldsymbol{y}_2^{\boldsymbol{s}} = \boldsymbol{G} \boldsymbol{q}^a_2 $.
\end{theorem}

\begin{proof}
Let $\boldsymbol{q}^a_1, \boldsymbol{q}^a_2 \in \mathcal{Q}^a$, and let $\boldsymbol{y}^s_i = \boldsymbol{G} \boldsymbol{q}^a_i$.
Now, assume that $\boldsymbol{q}^a_1 \neq \boldsymbol{q}^a_2$. Then, we have that
\begin{align}
\boldsymbol{y}^s_1 - \boldsymbol{y}^s_2 &= \boldsymbol{G} \boldsymbol{q}^a_1 - \boldsymbol{G} \boldsymbol{q}^a_2
= \boldsymbol{G} \underbrace{\left( \boldsymbol{q}^a_1 - \boldsymbol{q}^a_2 \right)}_{ = \boldsymbol{v}} = \boldsymbol{G} \boldsymbol{v}\\
&=  \sum_{i=1}^{n_r}  v_i \times \boldsymbol{g}_i 
= \sum_{i \in \mathcal{X}_{\boldsymbol{v}}} v_i \times \boldsymbol{g}_i \label{eqn:ys1-ys2}
\end{align}
where $\boldsymbol{g}_i$ is the $i^{\text{th}}$ column of $\boldsymbol{G}$.
\textbf{\textit{To show that $\boldsymbol{y}^s_1 - \boldsymbol{y}^s_2 \neq \boldsymbol{0}$, we need to show that all vectors $\boldsymbol{g}_i \; \forall i \in \mathcal{X}_{\boldsymbol{v}}$, are linearly independent.}} Otherwise, if such vectors $\boldsymbol{g}_i$ are linearly \textbf{\textit{dependent}}, then  $ \exists v_i \in \mathbb{R}$ for $i \in \mathcal{X}_{\boldsymbol{v}}$ such that $\boldsymbol{y}^s_1 - \boldsymbol{y}^s_2 = \boldsymbol{G} \boldsymbol{v} = \boldsymbol{0}$.

In fact, we can show a stronger statement: ``all vectors $\boldsymbol{g}_i \; \forall i \in \mathcal{X}_{\boldsymbol{q}^a_1} \cup \mathcal{X}_{\boldsymbol{q}^a_2} \supseteq \mathcal{X}_{\boldsymbol{v}}$, are linearly independent''.
Note that $\mathcal{X}_{\boldsymbol{q}^a_1}$ and $\mathcal{X}_{\boldsymbol{q}^a_2}$ are the sets of indexes of the non-zero components of $\boldsymbol{q^a_1}$ and $\boldsymbol{q^a_2}$, respectively (recall Definition \ref{def:setsOfIndexes}) and that
$\mathcal{X}_{\boldsymbol{q}^a_1} = \mathcal{X}_{{\boldsymbol{q}^a_{s1}}}$ and $\mathcal{X}_{\boldsymbol{q}^a_2} = \mathcal{X}_{{\boldsymbol{q}^a_{s2}}}$.

\begin{itemize}
\item First, let us show that $\mathcal{X}_{\boldsymbol{v}}$ is a subset of
$\mathcal{X}_{\boldsymbol{q}^a_1} \cup  \mathcal{X}_{\boldsymbol{q}^a_2}$.

Since $v_i = q^a_{1,i} - q^a_{2,i} \; \; \forall \; 1 \leq i \leq n_r$, then $q^a_{i,1} = q^a_{i,2} = 0 \Longrightarrow v_i = 0$.
Therefore, we have
\begin{align}
& \mathcal{X}_{\boldsymbol{q}^a_1}^c \cap  \mathcal{X}_{\boldsymbol{q}^a_2}^c  = \left\lbrace \mathcal{X}_{\boldsymbol{q}^a_1} \cup  \mathcal{X}_{\boldsymbol{q}^a_2} \right\rbrace ^c
\subseteq \mathcal{X}_{\boldsymbol{v}}^c
\end{align}
Then, by taking the complements of both sides we obtain the required result (note that $\{\cdot\}^c$ denotes a set complement).
\item Second, we show that vectors in the set $\mathcal{G} \triangleq \left\lbrace \boldsymbol{g}_i \vert i \in \mathcal{X}_{\boldsymbol{q}^a_1} \cup  \mathcal{X}_{\boldsymbol{q}^a_2} \right\rbrace$ are linearly independent over $\mathbb{F}_2$:
Assume towards contradiction that $\mathcal{G}$ is linearly dependent over $\mathbb{F}_2$. Hence, there exists a set $\mathcal{G}_D \subseteq \mathcal{G}$ such that any $\boldsymbol{g}_{i_0} \in \mathcal{G}_D$ can be written as a linear combination of all other vectors in $\mathcal{G}_D$, i.e.,
\begin{equation}
\boldsymbol{g}_{i_0} =
\sum_{\substack{ j :  j \neq i_0 \\\boldsymbol{g}_{j} \in \mathcal{G}_D}} \boldsymbol{g}_{j} \mod 2. \label{eqn:contradictionStatementOf_GD}
\end{equation}
Note that over $\mathbb{F}_2$, we can assume, without loss of generality (W.L.O.G.), that the coefficients of the linear combination above are $1$'s.
Hence, we have that 
\begin{equation}
\sum_{\substack{ j : \boldsymbol{g}_{j} \in \mathcal{G}_D}} \boldsymbol{g}_{j} \mod 2 = 0 
\end{equation}
Next, assume that $\exists \boldsymbol{q}^a_{s_3},\boldsymbol{q}^a_{s_4} {\in} \mathcal{Q}^a_s$ such that $\mathcal{X}_{\boldsymbol{v}_s}  {=} \left\lbrace j \vert \boldsymbol{g}_j \in \mathcal{G}_D \right\rbrace$ where $\boldsymbol{v}_s {=} \boldsymbol{q}^a_{s_3} {-} \boldsymbol{q}^a_{s_4} \mod 2$.

Then, since $\boldsymbol{G}$ is injective over $\mathcal{Q}^a_s$, then we have that
\begin{align}
\small
\boldsymbol{G} \boldsymbol{v}_s \mod 2  &= 
\sum_{\substack{ j  \in \mathcal{X}_{\boldsymbol{v}_s} }} \boldsymbol{g}_{j} \mod 2\\
&=
\sum_{\substack{ j : \boldsymbol{g}_{j} \in \mathcal{G}_D}} \boldsymbol{g}_{j} \mod 2
\neq \boldsymbol{0}\\
\xLongleftrightarrow{\textit{iff}}
&\; \boldsymbol{v_s} \mod 2 \neq \boldsymbol{0}
\end{align}
But, if $\mathcal{G}_D$ is non-empty, then $\boldsymbol{v}_s \neq 0$. Hence, we arrive at a contradiction to Eq. (\ref{eqn:contradictionStatementOf_GD}).
Therefore, the set $\mathcal{G}$ is linearly independent over $GF(2)$.
It remains to show that such $\boldsymbol{q}^a_{s_3}$ and $\boldsymbol{q}^a_{s_4}$ indeed exist.
Let us construct $\boldsymbol{q}^a_{s_3}$ as follows:\\
First, let $\boldsymbol{q}^a_{s_3} = \boldsymbol{q}^a_{s_1}$, then, reset its $i^\text{th}$ component to $0$ ($q^a_{s_3,i} = 0$) if $\boldsymbol{g}_i \not\in \mathcal{G}_D$.
Similarly, set $\boldsymbol{q}^a_{s_4} = \boldsymbol{q}^a_{s_2}$, then, reset the $i^\text{th}$ component to $0$ ($q^a_{s_4,i} = 0$) if $\boldsymbol{g}_i \not\in \mathcal{G}_D$ OR if $q^a_{s_1,i} = 1$.
Then, by construction, we have $\mathcal{X}_{\boldsymbol{q}^a_{s_3}} \subseteq \mathcal{X}_{\boldsymbol{q}^a_{s_1}}$ and $\mathcal{X}_{\boldsymbol{q}^a_{s_4}} \subseteq \mathcal{X}_{\boldsymbol{q}^a_{s_2}}$. Hence, by the inclusion property (recall Definition \ref{remark:inclusion}) we have $\boldsymbol{q}^a_{s_3} , \boldsymbol{q}^a_{s_4} \in \mathcal{Q}^a_s$ since both $\boldsymbol{q}^a_{s_1}, \boldsymbol{q}^a_{s_2} \in \mathcal{Q}^a_s$. Also, it is easy to see that $q^a_{s_3,j} - q^a_{s_4,j} \mod  2 = 1 \; \forall j : \boldsymbol{g}_j \in \mathcal{G}_D$.

\item Third, by Lemma \ref{lemma:vecIndOverFFAndR} below, the set $\mathcal{G}$, now taken over $\mathbb{R}$, is linearly independent.
\end{itemize}

Therefore, in Eq. (\ref{eqn:ys1-ys2}), it follows that $\boldsymbol{y}_1^{\boldsymbol{s}} - \boldsymbol{y}_2^{\boldsymbol{s}} \neq \boldsymbol{0}$ \textit{if and only if} $\boldsymbol{q}^a_1 - \boldsymbol{q}^a_2 \neq 0$, which concludes the proof.
\end{proof}

\begin{lemma}\label{lemma:vecIndOverFFAndR}
Any set of $n-$dimensional linearly independent vectors over $\mathbb{F}_2$ are also linearly independent over $\mathbb{C}$ if we interpret their $\mathsf{0}_{\mathbb{F}_2}$ and $1_{\mathbb{F}_2}$ components to be real scalars.
\end{lemma}
The proof is provided in Appendix \ref{append:FirstLemmaProof}.

\subsection{Beamforming Design}
Now, we focus our attention on the design of beamforming vectors $\boldsymbol{w_i}$, such that the measurement vector $\boldsymbol{y^s}$ is such that $\boldsymbol{y^s} = \boldsymbol{G} \boldsymbol{q^a}$.
Obviously, $\boldsymbol{w_i}$ depends on the chosen source code. Specifically, we want $\boldsymbol{w_i}$ to satisfy
\begin{equation} \label{eq:measurement_design_obj}
\boldsymbol{w_i}^H \boldsymbol{q} = \boldsymbol{y^s_i} = \boldsymbol{g_i} \boldsymbol{q^a}
= \sum_{j = 1}^{n_r} g_{i,j} q^a_j
\end{equation}
where $\boldsymbol{g_i}$ is the $i^{\text{th}}$ row of $\boldsymbol{G}$, and $g_{i,j}$ is the $j^{\text{th}}$ element of $\boldsymbol{g_i}$.
Recall that if $n_r$ antennas exist at RX, then there exists $n_r$ resolvable angular directions. Let us call these directions $d_1,d_2, \dots d_{n_r}$.
We want $\boldsymbol{w_i}$ to combine the received signal components at specific angular directions. Those angular directions are determined by $\boldsymbol{g_i} = \begin{pmatrix}
g_{i,1}, g_{i,2}, \dots, g_{i,n_r} \end{pmatrix}$.
Specifically, we want $\boldsymbol{w_i}$ to include the signal at directions $d_j$ for all $j$ such that $g_{i,j} = 1$.

Recall that $\boldsymbol{q} = \boldsymbol{U_r} \boldsymbol{q^a}$. Hence, we can rewrite Eq. (\ref{eq:measurement_design_obj}) as
$\boldsymbol{w_i}^H \boldsymbol{q} = \boldsymbol{w_i}^H \boldsymbol{U_r} \boldsymbol{q^a} = \boldsymbol{g_i} \boldsymbol{q^a}$.
Thus, we need to design $\boldsymbol{w_i}$ such that $\boldsymbol{w_i}^H \boldsymbol{U_r} = \boldsymbol{g_i}$, which can be rewritten as:
\begin{align}
\boldsymbol{w_i}^H \begin{pmatrix}
\boldsymbol{e_r} \left( 0 \right) &
\boldsymbol{e_r} \left( \frac{1}{L_r} \right) &
\hdots &
\boldsymbol{e_r} \left( \frac{n_r-1}{L_r} \right)\end{pmatrix} = \boldsymbol{g_i}.
\end{align}
Since the columns of $\boldsymbol{U_r}$ constitute an orthonormal basis, a very simple design of $\boldsymbol{w_i}$ is as a summation of the columns $\boldsymbol{e_r}\left(\frac{j-1}{L_r}\right)$ such that $g_{i,j} = 1$. In other words, we can design $\boldsymbol{w_i}$ as:
\begin{equation}
\boldsymbol{w_i} = \sum_{j: g_{i,j} = 1} \boldsymbol{e_r}\left( \frac{j-1}{L_r}\right)
\end{equation}

\begin{remark}
The adopted beamforming design is ideal under the perfect sparsity assumption (which we adopt). That is when channel paths lie at the angular directions defined in $\boldsymbol{U_r}$.
In practice, however, channel paths arrive at arbitrary angles in $[0, 2\pi]$. This makes each path contribute to multiple components in $\boldsymbol{q^a}$, hence, $\boldsymbol{q^a}$ is not perfectly sparse.
This happens due to (i) antenna side-lobes, and (ii) beam-overlap.
To resolve this problem, we can use side-lobe suppression techniques, e.g. Taylor Window, as well as large antenna arrays, which can form pencil-beam antenna patterns that avoid the beam-overlap problem. These come at the expense of a slight reduction in beam resolution.
We leave this investigation for a future study and only focus on the main idea of using source-coding-based measurements.
\end{remark}

\subsection{On the lower bound on the number of measurements}
\label{sec:LB}
In Theorem \ref{thm:uniquenessOfMeasurements}, we showed that a linear source code $C$ which can \textbf{uniquely} encode all $\boldsymbol{q}^a_s {\in} \mathcal{Q}^a_s$ can be used to design a framework that uniquely measures all $\boldsymbol{q}^a {\in} \mathcal{Q}^a$.
Let the compression ratio of the code $C$ be denoted by $r_c$ such that $r_c = \frac{m}{n_r}$.
where $m$ and $n_r$ are the number of rows and columns of $C$'s generator matrix $\boldsymbol{G}$, respectively.

Reducing the number of measurements is a fundamental objective for the mm-wave beam discovery problem. In light of Theorem \ref{thm:uniquenessOfMeasurements}, we can see that finding a source code with a high compression rate (small $r_c$) is crucial for attaining such an objective.
In the following discussion, we try to better understand the nature of this lower bound in the context of our proposed solution.
\begin{corollary}
\label{corr:LB}
Let $\ubar{m}$ denote the lowest possible number of measurements for mm-wave beam discovery using lossless, fixed-rate source coding. Then, we have
\begin{equation}
\ubar{m} \geq \ceil*{\log_2\left(\sum_{i = 0}^{L} {n_r \choose i}\right)} \geq H_2\left( \boldsymbol{q}^a_s \right)
\end{equation}
where $H_2\left(\cdot\right)$ is the binary entropy function.
\end{corollary}
\begin{proof}
Suppose that $C$ is a linear lossless fixed-rate source code which can uniquely compress all $\boldsymbol{q}^a_s \in \mathcal{Q}^a_s$.
By Theorem \ref{thm:uniquenessOfMeasurements}, we have that the number of measurements needed for estimating the mm-wave channel is equal to $m$ (the length of encoded channel support vectors).
Since the length of compressed sequences for any such code is lower bounded by $H_2\left( \boldsymbol{q}^a_s \right)$, then we have $\ubar{m} \geq H_2\left( \boldsymbol{q}^a_s \right)$. Moreover, since fixed-rate source codes do not take the probability distribution (i.e., frequency) of $\boldsymbol{q^a_s}$ into account, then we have
$\ubar{m} \geq \sup_{\mathbb{P}\left(\boldsymbol{q}^a_s \right)} H_2\left( \boldsymbol{q}^a_s \right) \geq H_2\left( \boldsymbol{q}^a_s \right), \text{  where}$
\begin{align}
\small \sup_{\mathbb{P}\left(\boldsymbol{q}^a_s \right)} \sum_{\boldsymbol{q^a} \in \mathcal{Q}^a} \mathbb{P} \left( \boldsymbol{q}^a_s \right) \log_2\left( \frac{1}{\mathbb{P} \left( \boldsymbol{q}^a_s \right)} \right)
\small  &= \log_2\left(\vert\mathcal{Q}^a_s\vert\right) \label{eq:supsol}
\end{align}
The result of solving the $\sup$ problem in Eq. (\ref{eq:supsol}) is $\mathbb{P}\left(\boldsymbol{q}^a_s \right) = \frac{1}{\vert \mathcal{Q}^a_s \vert} \; \forall \boldsymbol{q}^a_s$ since the uniform distribution maximizes the entropy.
Since the number of measurement has to be an integer, we take the ceil of right hand side of Eq. (\ref{eq:supsol}).
Finally, by the inclusion property in Definition \ref{remark:inclusion}, we have $\vert \mathcal{Q}^a_s \vert = \sum_{i = 0}^L {n_r \choose i}$, which concludes the proof.
\end{proof}

\subsection{Channel Estimation Error}
\label{sec:errorAnalysis}
In Theorem \ref{thm:uniquenessOfMeasurements}, we have shown how to obtain unique measurements $\boldsymbol{y^s}\; \forall\; \boldsymbol{q^a}{\in} \mathcal{Q}^a$.
Recall that $\boldsymbol{y^s}{=}\boldsymbol{G} \boldsymbol{q^a}$ is an error-free measurement vector.
In practice, however, measurements are never error-free. Measurements errors are bound to happen due to the effects of \textbf{\textit{thermal and quantization noise}}, among others factors. The quality of channel estimates obtained using error-corrupted measurements is essentially degraded, which calls for a deeper understanding of the effects of such errors.
A crucial question we try to answer here is: \textbf{\textit{Do small perturbations/imperfections in channel measurements make channel estimates considerably deviate from their true values?}}
In this section, we shed some light on this problem by deriving an upper bound on channel estimation error as a function of measurement error.
We also show that for a special class of generator matrices, the channel estimation error, measured using the $\ell_2-$norm, is smaller than or equal to the $\ell_2$ norm of the measurement error.

We denote error-corrupted measurements using $\boldsymbol{u^s}$ such that
\begin{equation}
\boldsymbol{u^s} = \boldsymbol{y^s} + \boldsymbol{z},
\end{equation}
where $\boldsymbol{z}$ is the measurement error (recall Eq. (\ref{eqn:noisyMeasurements}) and the discussion that follows).
Assume that we can \textbf{perfectly} decode any measurement vector into its corresponding channel. That is, for any measurement vector $\boldsymbol{y^s}$, we can find a corresponding $\boldsymbol{q^a}$ such that $\boldsymbol{y^s} {=} \boldsymbol{G} \boldsymbol{q^a}$ (measurement decoding will be further discussed in Section \ref{sec:DNN}).
Let us also denote the \textit{channel estimate} obtained using error-corrupted measurements $\boldsymbol{u^s}$ by $\hat{\boldsymbol{q}}^{\boldsymbol{a}}$, i.e., $\boldsymbol{u^s} {=} \boldsymbol{G} \hat{\boldsymbol{q}}^{\boldsymbol{a}}$.
The following proposition provides an upper bound on the channel estimation error in terms of measurements errors.

\begin{proposition}\label{thm:estimationErrorBound}
Assume perfect measurement decoding, and let $\sigma_{\text{min}}\left( \cdot \right)$ denote the minimum singular value of a given matrix. Then, the channel estimation error is upper bounded as:
\begin{equation}
\norm{\hat{\boldsymbol{q}}^{\boldsymbol{a}} - \boldsymbol{q^a}}_2
\leq
\frac{1}{\sigma_{\text{min}}\left(\boldsymbol{G}\right)} \norm{\boldsymbol{z}}_2
\end{equation}
\end{proposition}
\begin{proof}
Let us start by writing $\boldsymbol{z}$ as:
$\boldsymbol{z} = \boldsymbol{u^s} - \boldsymbol{y^s} = \boldsymbol{G} \left( \hat{\boldsymbol{q}}^{\boldsymbol{a}} - \boldsymbol{q^a} \right)$. Therefore, we have
\begin{align}
\small
\Longrightarrow \norm{\boldsymbol{z}}_2 &= \norm{\boldsymbol{G} \left( \hat{\boldsymbol{q}}^{\boldsymbol{a}} - \boldsymbol{q^a} \right)}_2\\
&= \norm{\left( \hat{\boldsymbol{q}}^{\boldsymbol{a}} - \boldsymbol{q^a} \right)}_2 
\frac{ \norm{\boldsymbol{G} \left( \hat{\boldsymbol{q}}^{\boldsymbol{a}} - \boldsymbol{q^a} \right)}_2}{\norm{\left( \hat{\boldsymbol{q}}^{\boldsymbol{a}} - \boldsymbol{q^a} \right)}_2}\\
&\geq  \norm{\left( \hat{\boldsymbol{q}}^{\boldsymbol{a}} - \boldsymbol{q^a} \right)}_2  \sigma_{\text{min}}\left( \boldsymbol{G} \right) \label{eq:rearrange}
\end{align}
Finally, by rearranging (\ref{eq:rearrange}), we obtain the required statement
\begin{equation*}
\norm{\hat{\boldsymbol{q}}^{\boldsymbol{a}} - \boldsymbol{q^a}}_2
\leq
\frac{1}{\sigma_{\text{min}}\left(\boldsymbol{G}\right)} \norm{\boldsymbol{z}}_2 \qedhere
\end{equation*}
\end{proof}
Now, we see that if $\sigma_{\text{min}}\left(\boldsymbol{G}\right) {\geq} 1$, then the channel estimation error (measured using the $\ell_2-$norm) is smaller than or equal to the $\ell_2-$norm of the measurement error, i.e., $\norm{\hat{\boldsymbol{q}}^{\boldsymbol{a}} - \boldsymbol{q^a}}_2 \leq \norm{\boldsymbol{z}}_2$. 
This, in fact, is the case for the class of generator matrices introduced in the following proposition
\begin{proposition}
\label{prop:minSingularValStdForm}
Let $\boldsymbol{I}_{m }$ be the $m{\times}m$ identity matrix. Then, $\sigma_{\text{min}}\left(\boldsymbol{G}\right) {\geq} 1$ for $\boldsymbol{G}$ of the form:
\begin{equation}
\label{eqn:minSingularValStdForm}
\boldsymbol{G} = 
\begin{pmatrix}[c|c]
\boldsymbol{I}_{m} & \boldsymbol{P}_{m \times n-m}
\end{pmatrix}
\end{equation}
\end{proposition}
See Appendix \ref{append:SecondPropositionProof} for proof.
\begin{remark}
It is not difficult to obtain generator matrices of the form in Eq. (\ref{eqn:minSingularValStdForm}). For instance, syndrome source codes can be manipulated using row and column operations over the binary field to produce equivalent codes with $\boldsymbol{G}$ as in Eq. (\ref{eqn:minSingularValStdForm}).
\end{remark}

\section{Measurement Decoding}
\label{sec:DNN}
\begin{figure*}[t]
\centering
\captionsetup[subfigure]{}
\begin{subfigure}{.33\textwidth}
  \centering
  \captionsetup{justification=centering}
\includegraphics[width=0.95\linewidth]{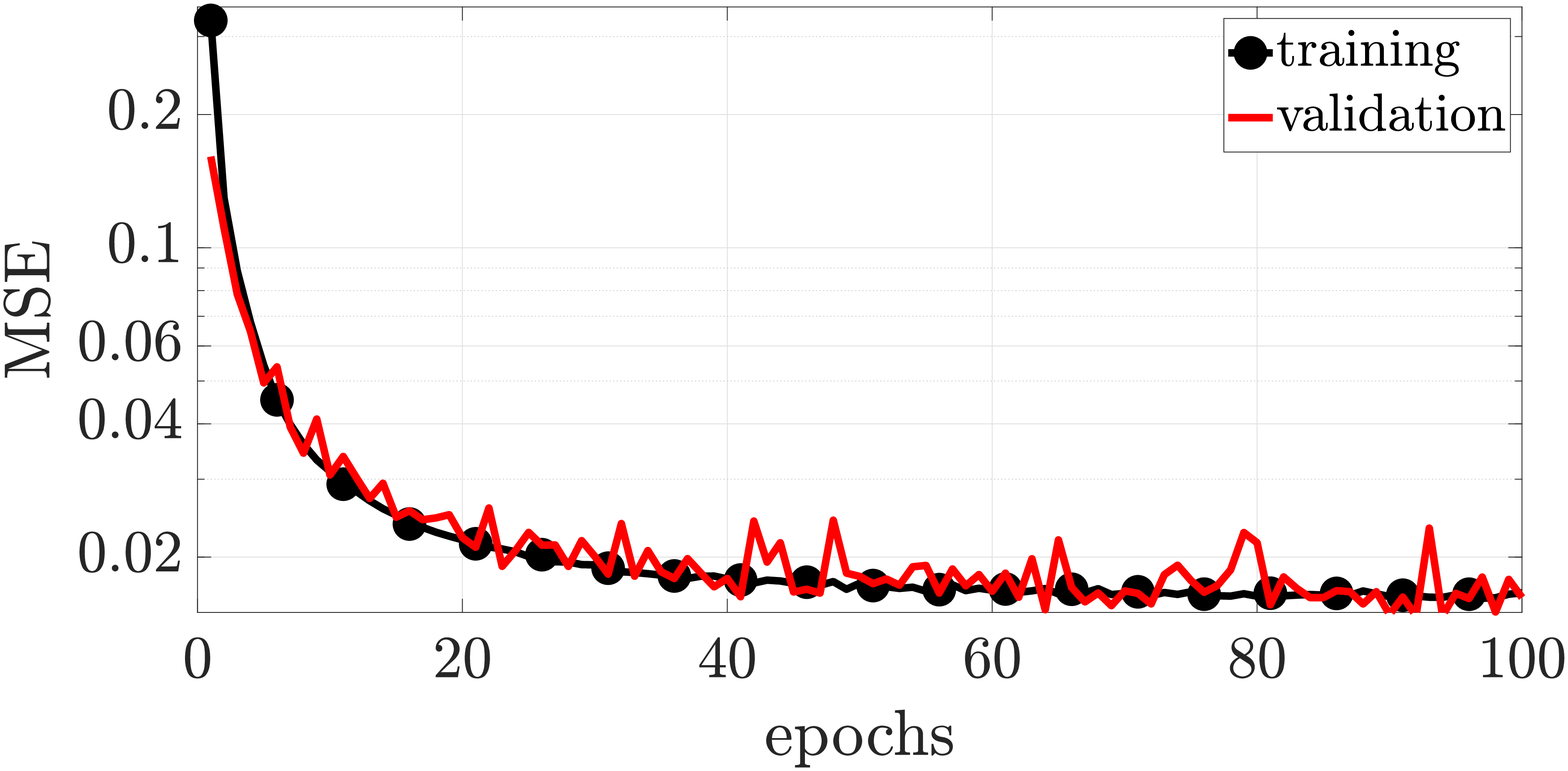}
\caption{\small Training vs. validation loss.}
\label{fig:trainingPerformance} 
\end{subfigure}%
\begin{subfigure}{.33\textwidth}
  \centering
  \captionsetup{justification=centering}
\includegraphics[width=0.95\linewidth]{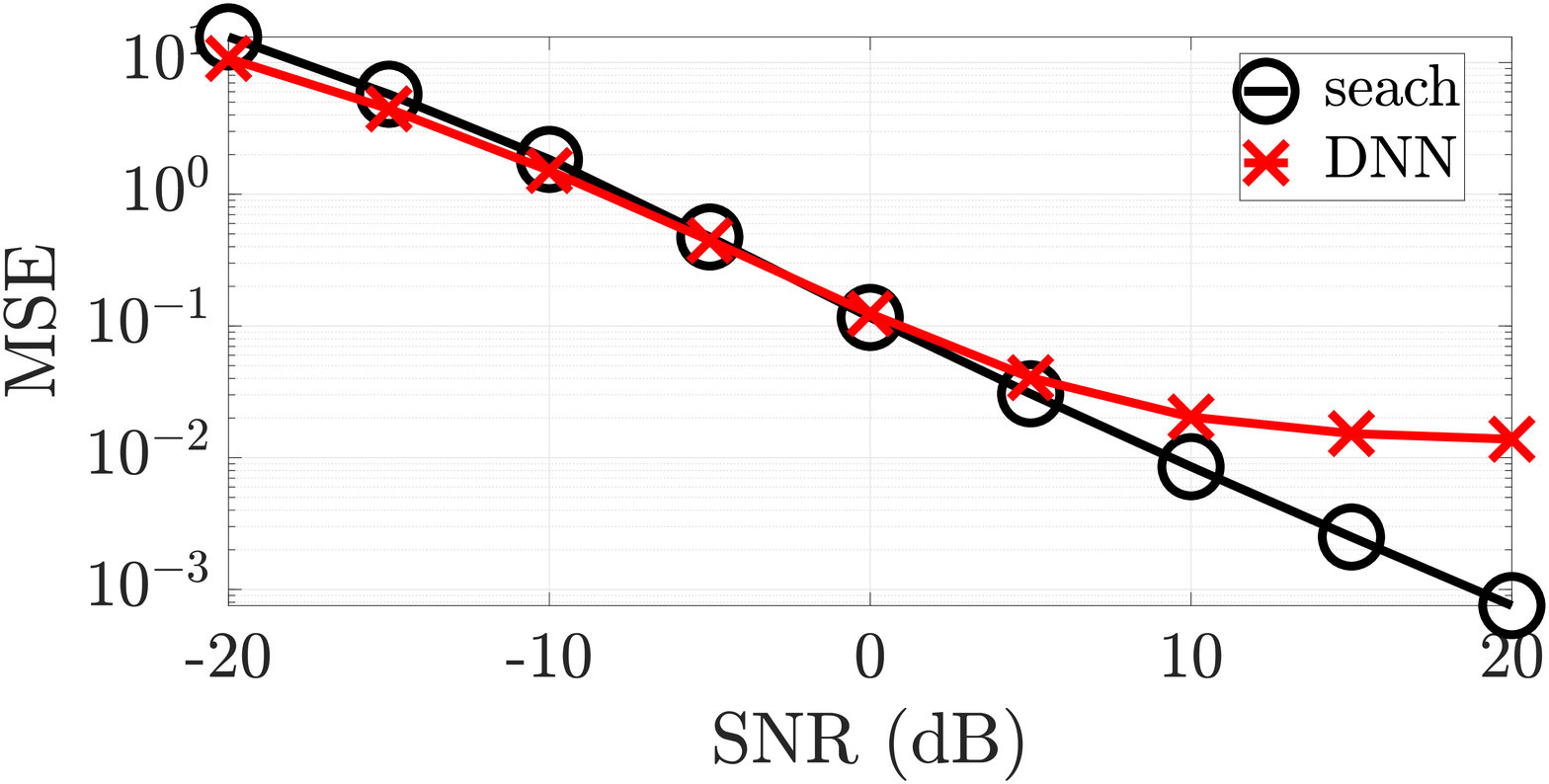}
\caption{\small Mean Squared Error (MSE).}
\label{fig:loss_DNNvsSearch} 
\end{subfigure}%
\begin{subfigure}{.33\textwidth}
  \centering
  \captionsetup{justification=centering}
\includegraphics[width=0.95\linewidth]{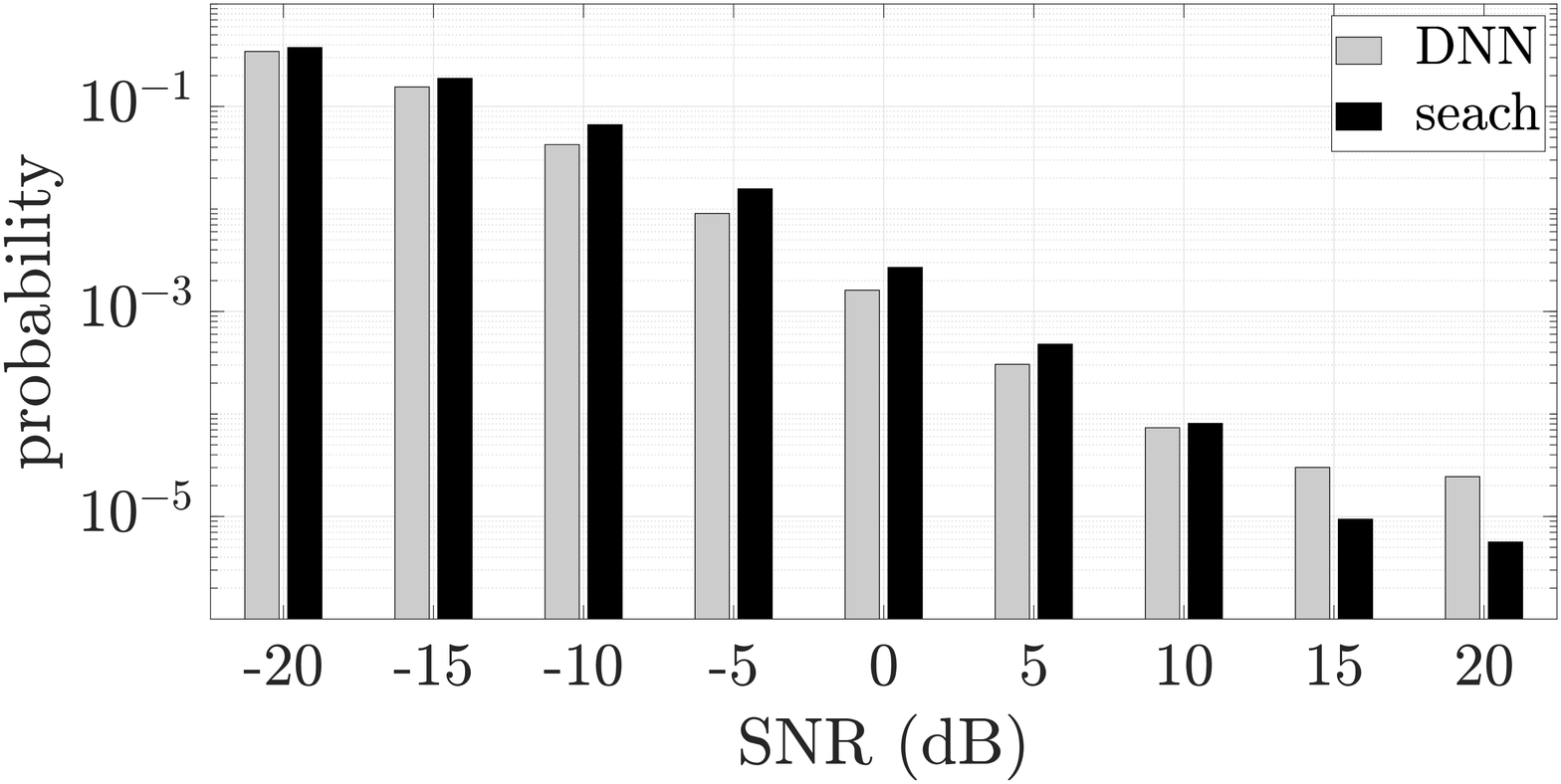}
\caption{\small Probability of path misdetection.}
\label{fig:k_ge1_DNNvsSearch}
\end{subfigure}%
\caption{Evaluation of DNN-based measurement to channel mapping.}
\label{fig:beamDetectionProb_DNNvsSearch}
\end{figure*}
Designing channel measurements that have one-to-one correspondence with $\boldsymbol{q}^a$ is only part of the solution.
Equally important, however, is the ability to ``decode'' $\boldsymbol{y^s}$ back to $\boldsymbol{q}^a$, i.e., figuring out what the function $\mathcal{D}(\cdot)$, in Eq. (\ref{eqn:decodingFn}), is.
The one-to-one correspondence between $\boldsymbol{q}^a$ and $\boldsymbol{y^s}$ guarantees that there exists an inverse function that maps $\boldsymbol{y^s}$ back to $\boldsymbol{q^a}$. Nevertheless, since we can only obtain $\boldsymbol{u^s}$; an error-corrupted version of $\boldsymbol{y^s}$, we cannot exactly regenerate $\boldsymbol{q^a}$, but rather, an estimate $\hat{\boldsymbol{q}}^a$.
Given that measurement errors occur, our objective is to obtain $\hat{\boldsymbol{q}}^a$ such that its \textbf{\textit{distance}} to $\boldsymbol{q^a}$ is as small as possible (i.e., minimize the estimation error).
We use the $l_2$ norm as a distance measure between $\boldsymbol{q}^a$ and $\hat{\boldsymbol{q}}^a$, defined as $\delta \left( \boldsymbol{q^a} , \hat{\boldsymbol{q}}^{\boldsymbol{a}} \right) \triangleq \norm{\boldsymbol{q^a} - \hat{\boldsymbol{q}}^{\boldsymbol{a}}}_2$

Optimal measurement decoding requires solving an $\ell_0$-norm minimization problem \cite{CS4Wireless_TipsAndTricks}.
This problem is not convex and its solution requires heavy computations, which is intractable for channels with large dimensions and/or relatively high sparsity level.
An example for optimal measurement decoding is the ``search'' decoding, proposed in \cite{shabaraBeamToN}, which requires a combinatorial search over the column subspaces of $\boldsymbol{G}$, and whose complexity is of order $O\left(n_r^L\right)$.
Again, this is prohibitive for large antenna arrays ($n_r$) and large $L$.
Another solution is the ``look-up'' table method in \cite{shabaraBeamToN}, where quantized measurement-channel pairs are stored in memory. Here, the channel vector whose corresponding stored measurement is closest to the collected measurement is selected. The main disadvantage with this method is that the table size increases dramatically with the number of measuremtns and ADC quantization resolution.
\textbf{Motivated by the drawbacks of the look-up table and search methods, we propose an alternative \textit{Machine Learning (ML)} based approach that uses \textit{Deep Neural Networks (DNN)}.} DNNs in particular are commonly used as function approximation algorithms, hence they provide an appealing light-weight, data-driven, alternative solution for the measurement decoding problem.
Our main goal here is to reduce the computational complexity while still maintaining reliable measurement decoding.

\subsection{DNN-based mapping}
\label{sec:DnnBasedMapping}
ML is widely used to solve very complex problems through \textit{learning}. We focus on \textit{supervised learning} to solve the decoding problem,
which is a multi-dimensional non-linear regression problem for which neural networks is a powerful tool.
Specifically, we use a fully connected classical DNN with an input layer of $m$ nodes (equal to the measurement dimensions) and an output layer of $n_r$ nodes (equal to the channel dimensions).
The DNN model is designed to handle real-valued input-output data. But, on the contrary, both the channel $\boldsymbol{q^a}$ and measurements $\boldsymbol{y^s}$ are complex-valued.
To overcome this problem, observe that $\boldsymbol{y^s}$ can be written as $\boldsymbol{y}_R^{\boldsymbol{s}} {+} j \boldsymbol{y}_I^{\boldsymbol{s}}$ and $\boldsymbol{q^a}$ as $\boldsymbol{q}_R^{\boldsymbol{a}} {+} j \boldsymbol{q}_I^{\boldsymbol{a}}$ (i.e., in terms of their real and imaginary components).
And notice that $\boldsymbol{y}_R^{\boldsymbol{s}} {=} \boldsymbol{G} \boldsymbol{q}_R^{\boldsymbol{a}}$ and $\boldsymbol{y}_I^{\boldsymbol{s}} {=} \boldsymbol{G}  \boldsymbol{q}_I^{\boldsymbol{a}}$.
Therefore, we can construct an estimate $\hat{\boldsymbol{q}}^{\boldsymbol{a}}$ using its real and imaginary components, i.e., $\hat{\boldsymbol{q}}^{\boldsymbol{a}}_R {+} j \hat{\boldsymbol{q}}^{\boldsymbol{a}}_I$, where $\hat{\boldsymbol{q}}^{\boldsymbol{a}}_R$ and $\hat{\boldsymbol{q}}^{\boldsymbol{a}}_I$ are estimated using $\boldsymbol{y}_R^{\boldsymbol{s}}$ and $\boldsymbol{y}_I^{\boldsymbol{s}}$ as inputs to the DNN model, respectively.
Therefore, our DNN takes the measurement vectors $\boldsymbol{y^s}$ as inputs and produces the corresponding channel estimates $\hat{\boldsymbol{q}}^{\boldsymbol{a}}$ at its output, but it does so in two different steps, handling the real and imaginary parts separately. For ease of notation, we will not use real and imaginary components to refer to inputs and outputs of the DNN models but it should be understood that this is how we handle it.
The number of hidden layers and their corresponding number of nodes are design parameters that depend on the sizes of the input and output, and the relationship governing them.
For all hidden layers, we use the rectified linear (ReLU) activation function while for the output layer we use the linear activation function.
We also use the \textit{ADAM} optimizer \cite{adam} for training and the Mean Squared Error (MSE) loss function to quantify the model error\footnote{We use Keras API \cite{chollet2015keras} to build, train, test and use the DNN model we propose. Our pre-trained models and DNN-related codes are available in \cite{simCodes}.}.

\textbf{Model training:}
Although we do not have a closed form expression for mapping $\boldsymbol{y^s}$ to $\boldsymbol{q^a}$ (hence the need for an algorithmic solution), generating training data is actually straightforward. This is because the reverse direction (i.e., mapping $\boldsymbol{q^a}$ to $\boldsymbol{y^s}$) is just a simple linear transformation.
Training data is generated as follows: For every $\boldsymbol{q}^a_s \in \mathcal{Q}^a_s$ (recall that $\mathcal{Q}^a_s$ is the set of all possible channel support vectors defined in Section \ref{sec:LB_mmwave}), we generate $n_s = 300$ random channels, $\boldsymbol{q}^a$, by choosing the non-zero components of $\boldsymbol{q^a}$ to be uniformly distributed in $[-\alpha^b_{\text{max}},\alpha^b_{\text{max}}]$ where $\alpha^b_{\text{max}}$ (recall Eq. (\ref{eqn:basebandPathGain})) is the maximum magnitude of baseband path gains, which can be obtained using channel statistics. Note that we can set $\alpha^b_{\text{max}}$ to be the maximum ADC quantized value of $\vert \alpha_l^b\vert$.
Thus, the total number of input-output samples we have is $n_s \times \abs{\mathcal{Q}^a_s}$.
We use $70\%$ of these samples for training and the remaining $30\%$ for validation. Training is done using $200$ epochs with batches of size $32$. We monitor the validation error to make sure that the model does not over-fit the training data. If over-fitting is observed (which is indicated by a persistent increase in validation error at the end of every epoch), we stop the training process and only keep the model which produced the least validation error.
DNN training is done offline, and a trained DNN model is stored in memory to be used when needed.

\begin{figure*}[t]
\centering
\includegraphics[width=0.9\linewidth]{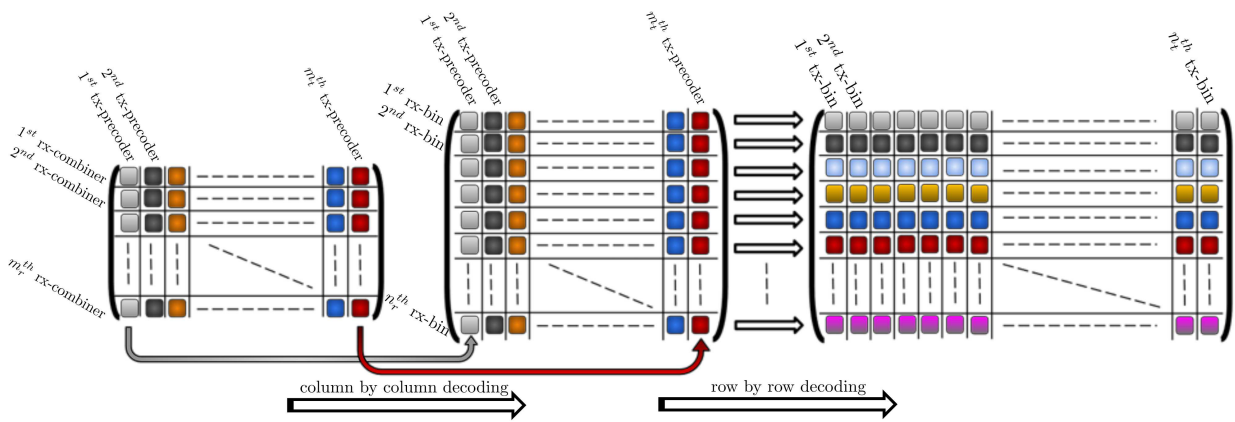}
\caption{\small Measurement decoding for channels with multiple TX/RX antennas is done in two steps.
Given the matrix $\boldsymbol{Y^s}$ whose $y^s_{i,j}th$ component is $= \boldsymbol{w}_i^H \boldsymbol{Q} \boldsymbol{f}_j$ (shown on the left), we first do a column by column decoding where the $j^{th}$ decoded column of $\boldsymbol{Y^s}$ is $\boldsymbol{q}^a_{rx,j}$.
Then, in the second step we decode the intermediary matrix row by row to produce $\hat{\boldsymbol{Q}}^{\boldsymbol{a}}$ (shown on the right).
The two decoding functions of first and second steps are dependent on the source codes used to design $\boldsymbol{w_i}$'s and $\boldsymbol{f_j}$'s, respectively.}
\label{fig:multiTxRx}
\end{figure*}

\subsection{DNN Model Assessment}
\label{sec:modelAssessment}
To argue the reliability of DNN-based mapping, we test it using a channel with $n_t {=} 1$, $n_r {=} 23$ and a maximum of $3$ paths i.e., $L {\leq} 3$. 
We compare its performance against the ``search'' method of \cite{shabaraBeamToN}.
Based on the described channel parameters, only $m {=} 11$ measurements are needed to discover its paths (more details about this particular example are discussed in Section \ref{sec:eval}).
We design a DNN model with an input layer of $m{=} 11$ nodes and an output layer of $n_r {=} 23$ nodes. The model also has $5$ hidden layers with $1024, 512, 512, 128$ and $128$ nodes, respectively.
We train the DNN model using data generated as described in Section \ref{sec:DnnBasedMapping}.
Fig. \ref{fig:trainingPerformance} shows the average MSE loss of both training and validation data sets for $100$ epochs.
Training achieves validation error of ${\approx} 0.0143$ (averaged across all samples of validation data).
The figure also shows close MSE values for training and validation. This indicates that the model generalizes well to measurements it had not seen before, which guarantees reliability for arbitrary measurements.

These initial results are promising. However, they are obtained using error-free measurements. This prompts us to test the resilience of DNN-based mapping against noisy measurements. We also compare its performance against the search method proposed in \cite{shabaraBeamToN}.
To do so, we generate a testing data set in the same way we created the training data. We also generate sets of uniformly distributed noise vectors where each noise set is drawn at a different value of transmit SNR from $-20$ to $20$ dB.
The noise vectors are then added to the inputs (channel measurements) of the testing data set then passed through the trained model. The decoded $\hat{\boldsymbol{q}}^{\boldsymbol{a}}$ is recorded at the output.
Similarly, we use the ``search'' method to decode the same noise-corrupted measurements.

For evaluation, we use i)  average MSE, as well as ii) the probability of path misdetection (i.e., no path discovery).
We say a path is correctly discovered if the path gain of its corresponding component in $\hat{\boldsymbol{q}}^{\boldsymbol{a}}$ is among the $L {=} 3$ strongest components in $\hat{\boldsymbol{q}}^{\boldsymbol{a}}$.
Fig. \ref{fig:loss_DNNvsSearch} shows the average MSE obtained using the search and DNN-based mapping methods on a log scale.
We see that at low SNR, DNN-based decoding outperforms the search method. This indicates that the DNN model is more resilient against measurement errors. At high SNR, however, the DNN's MSE saturates at ${\approx} 0.014$ which is the same value we obtain for validation during model training using noise-free inputs (not that the MSE value at which DNN-based mapping saturates can be made lower by further improvement of the DNN model).
The search method's MSE, on the other hand, keeps improving as SNR increases, nevertheless, for values below $10^{-2}$ the improvement is marginal.
The probability of path misdetection, shown in Fig. \ref{fig:k_ge1_DNNvsSearch}, confirms the performance trend of the MSE. Specifically, we see that at low SNR, the DNN-based model outperforms the search method (i.e., has lower probability of misdetection) while at high SNR we see that the search method is better.

\textbf{Computational complexity:}
As we have previously discussed, the search method requires high computational power. Precisely, ${n_r \choose L}$ iterations with one matrix inversion and two matrix multiplication operations are performed per iteration, which then produces a vector of length $m$. Finally, an additional step of finding the minimum $l_2-$norm of all $n_r \choose L$ vectors is performed.
On the other hand, the DNN-based mapping just requires $N_{k}$ linear computations for the hidden and output layers, where $N_k$ is the number of nodes at the $k^{\text{th}}$ layer. These computations are of the form $\sum_{i=1}^{N_{k-1}} w_i a_i$ where $a_i$ is the value passed from the $i^{th}$ node of the previous layer and $w_i$ is the weight on its link. 
For this particular example, the search method and the DNN-based method were implemented on the same machine and on average the search method's execution time was $11.2$ ms compared to $47 \; \mu$s for the DNN model.

\section{Multiple Transmit and Receive Antennas}
\label{sec:multiTxRx}
So far, we only dealt with channels of single-transmit, multiple-receive antennas. Recall that this setting is almost identical to multiple-transmit, single-receive antenna channels, 
except that in the former setting we seek to design $\boldsymbol{w_i}$'s to estimate the angular channel at RX, while in the latter, we design $\boldsymbol{f_j}$'s to estimate the angular channel at TX.
In this section, we extend the channel setting to be of \textbf{\textit{multiple-transmit, multiple-receive}} antennas.
We build on the design principles and decoding methods of single transmit antenna channels and show how measurements are obtained and decoded to estimate the entire $n_r {\times} n_t$ channel.

\subsection{Measurements}
\label{sec:multipleTXRXMeasurements}
Unlike the single transmit antenna scenario where TX sends signals omnidirectionally, it can now focus its transmission on narrow angular directions.
However, from RX's point of view, no matter which set of directions the TX is transmitting into, it can only see a number of $n_r$ resolvable bins; only $L$ of which may have paths to TX.
The same is true from TX's perspective, where the TX can only see $n_t$ resolvable bins, only $L$ of which may have paths to the receiver\footnote{Recall that the directions at which the TX is transmitting and the RX is receiving are determined by their antenna beam patterns which are in turn determined by $\boldsymbol{f_j}$ and $\boldsymbol{w_i}$, respectively (see Fig. \ref{fig:beamforming_demo}).}.
Thus, for an arbitrary tx-precoder, the receiver would need to measure the channel using the same set of $\boldsymbol{w_i}$'s it needs for the $n_t = 1$ setting. Upon decoding, the result would be $n_r$ angular rx bins (corresponding to the particular $\boldsymbol{f_j}$ used at TX).
Similarly, for an arbitrary rx-combiner, the transmitter would need the same set of $\boldsymbol{f_j}$'s it needs for the $n_r = 1$ setting to find its respective tx bins.
To find such $\boldsymbol{f_j}$'s and $\boldsymbol{w_i}$'s, we invoke Theorem \ref{thm:uniquenessOfMeasurements}.

Let $\boldsymbol{f}_j \; \forall j {\in} \{ 1, \dots, m_t \}$ be the tx-precoding vectors and $\boldsymbol{w}_i \; \forall i {\in} \{1, \dots, m_r\}$ be the rx-combining vectors.
Then, the channel measurements are obtained as follows:
The transmitter sends a number of $m_r$ pilot symbols using \textbf{\textit{each}} of its $m_t$ precoders.
On the receiver side, for every tx-precoder, $m_r$ channel measurements are obtained using the distinct $m_r$ rx-combiners.
Recall that $u_{i,j} = y^s_{i,j} + \boldsymbol{w}_i^H \boldsymbol{n}$
where $y^s_{i,j} = \boldsymbol{w}_i^H \boldsymbol{Q} \boldsymbol{f}_j$ (see Eq. (\ref{eqn:error-free_symbols})).
Let us arrange the $m_r$ measurements corresponding to the $j^{th}$ tx-precoder in $\boldsymbol{y}^s_j$ and define $\boldsymbol{Y}^s$ as
\begin{equation}
\boldsymbol{Y}^s \triangleq
\begin{pmatrix}
\boldsymbol{y}^s_1 & \boldsymbol{y}^s_2 & \dots &
\boldsymbol{y}^s_{m_t}
\end{pmatrix}
\end{equation}
Thus, $\boldsymbol{Y^s}$ contains all $m_t{\times}m_r$ channel measurements necessary to discover all available paths.

\subsection{Decoding $\boldsymbol{Y}^s$}

To obtain $\hat{\boldsymbol{Q}}^a$ from $\boldsymbol{Y}^s$, we perform multiple SIMO decoding operations\footnote{Alternatively, we could have trained a large DNN model which accepts all measurements $\boldsymbol{Y^s}$ and outputs an estimate $\hat{\boldsymbol{Q}}^{\boldsymbol{a}}$. Adopting this strategy, however, results in overwhelming training complexity since this model would need to be trained with a massive training data set of size $n_s \times |\mathcal{Q}^a_s| = n_s \sum_{i = 0}^L {n_r \times n_t \choose i}$, where $\mathcal{Q}^a_s$ now is the set of all support vectors of size $n_r n_t {\times} 1$ that represent the $n_r {\times} n_t$ vectorized channel matrices.
Compare this to our solution of using 2 DNN models trained with data sets of sizes $n_s \sum_{i = 0}^L {n_r \choose i}$ and $n_s \sum_{i = 0}^L {n_t \choose i}$, respectively.},
as described in Section \ref{sec:DNN}.
This procedure is highlighted in the diagram in Fig. \ref{fig:multiTxRx} and is detailed as follows:
\begin{enumerate}[(i)]
\item Decode every $\boldsymbol{y}^s_j \; \forall j \{1,2,\dots,m_t\}$ to obtain $\boldsymbol{q}^a_{rx,j}$.
Recall that $\boldsymbol{y}^s_j$ is the measurement vector corresponding to the $j^{th}$ tx-precoder.
Thus, $\boldsymbol{q}^a_{rx,j}$, is the $n_r {\times} 1$ mm-wave channel observed at RX due to the TX signal transmission through the angular directions featured by $\boldsymbol{f}_j$.
\item After Step (i), we obtain a sequence of $m_t$ ``measurement'' components corresponding to each rx-bin.
Each of these components is produced using a distinct tx-precoder. Let us denote these sequences by $\boldsymbol{y}^s_{tx,k}$ ($1 {\times} m_t$ row vectors), where $k \in \{1,2,\dots,n_r\}$.
\item Decode each $\boldsymbol{y}^s_{tx,k}$ to obtain $\boldsymbol{q}^a_{tx,k}$ ($1 {\times} n_t$ row vectors) whose components constitute all the tx-bins corresponding to the $k^{th}$ rx-bin.
\item Stack all $\boldsymbol{q}^a_{tx,k}$ to obtain $\hat{\boldsymbol{Q}}^a$ (each representing the $k^\text{th}$ row of $\hat{\boldsymbol{Q}}^a$).
\end{enumerate}

\section{Performance Evaluation}
\label{sec:eval}
We evaluate the performance of our proposed coding-based solution under various simulation scenarios.
Specifically, we consider $23 {\times} 23$ and $15 {\times} 32$ multi-path channels.
For both channel settings, we assume the existence of a maximum of $3$ paths\footnote{Note that information on $L$ can be obtained from statistical channel models \cite{akdeniz2014millimeter, sur2016beamspy}).}.
We also consider $15{\times} 31$ single-path channels.
The single-path assumption is appropriate for LoS scenarios where the path gain of the LoS is significantly higher than the gains of Non-LoS (NLoS) paths ($\approx20$ dB higher \cite{akdeniz2014millimeter}).

To understand how our solution compares to the state-of-the-art, we implement a compressed-sensing-based channel estimation solution, as well as the IEEE 802.11ad's (WiGig) beam discovery method.
Note that, while compressed sensing is a generic solution that can be applied to multi-path channels (similar to our solution), the WiGig method is designed to discover one channel path, hence, we only use it for the $15\times 31$ single-path channel.
Our results demonstrate that our proposed solution is more resilient to errors compared to both CS and 802.11ad, and produces higher quality estimates.
Furthermore, we study the effect of ADC resolution on channel estimation performance.
This is important because ADC's power consumption is directly proportional to their resolution.
Hence, it is necessary to understand the resolution limit beyond which only minimal gains, in terms of channel estimation performance, can be achieved.

\begin{figure}[t]
\centering
\includegraphics[width=0.45\linewidth]{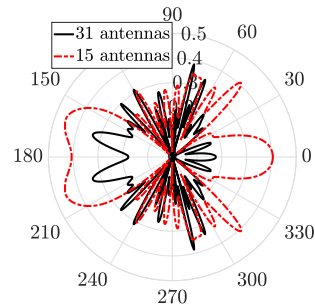}
\caption{\footnotesize Antenna pattern example using CS.}
\label{fig:rx_CS} 
\end{figure}

\subsection{Performance Metrics}
We adopt performance metrics that highlight:
(i) the measurement overhead,
(ii) accuracy of path discovery
(iii) quality of the estimated path gains, and
(iv) the impact of the channel estimates on achievable data rate.
These metrics are evaluated numerically, using Monte Carlo simulations (averaged over $10^5$ simulation runs) and evaluated against different values of SNR.
We define the performance metrics as follows:
\textbf{\textit{i) Number of measurements}}: Given by $m_t{\times}m_r$.
\textbf{\textit{ii) Probability of path discovery:}} Paths are said to exist at the strongest $L$ components in the estimated channel $\hat{\boldsymbol{Q}^a}$.
\textbf{\textit{iii) MSE}} (Normalized): Defined as the squared value of the Frobenius-norm of the estimation error, $\boldsymbol{Q^a} - \hat{\boldsymbol{Q}}^{\boldsymbol{a}}$, normalized by the Frobenius-norm of $\boldsymbol{Q^a}$, i.e., $\small \frac{\norm{\boldsymbol{Q^a} - \hat{\boldsymbol{Q}}^{\boldsymbol{a}}}_F^2}{\norm{\boldsymbol{Q^a}}_F^2}$.
\textbf{\textit{iv) Outage Rate}}: Denoted by $R_{\text{out}}$ and defined as
$ R_{\text{out}} \triangleq \mathbb{E}\left[ \left( 1- \mathbbm{1}_{\{ \text{out}\}} \right) \times C_{\boldsymbol{Q}} \right]$
where $C_{\boldsymbol{Q}}$ is the MIMO channel capacity of the channel $\boldsymbol{Q}$ \cite{tse2005fundamentals},
and $\mathbbm{1}_{\{ \text{out}\}}$ is the indicator function that takes a value of $1$ in case of outage and $0$ otherwise. We assume that an outage occurs if any of the strong channel paths were misdetected.

\begin{figure*}[t]
\centering
\captionsetup[subfigure]{aboveskip=-0.25pt,belowskip=-0.25pt}
\begin{subfigure}[t]{0.33\textwidth}
\centering
\includegraphics[width=0.95\linewidth]{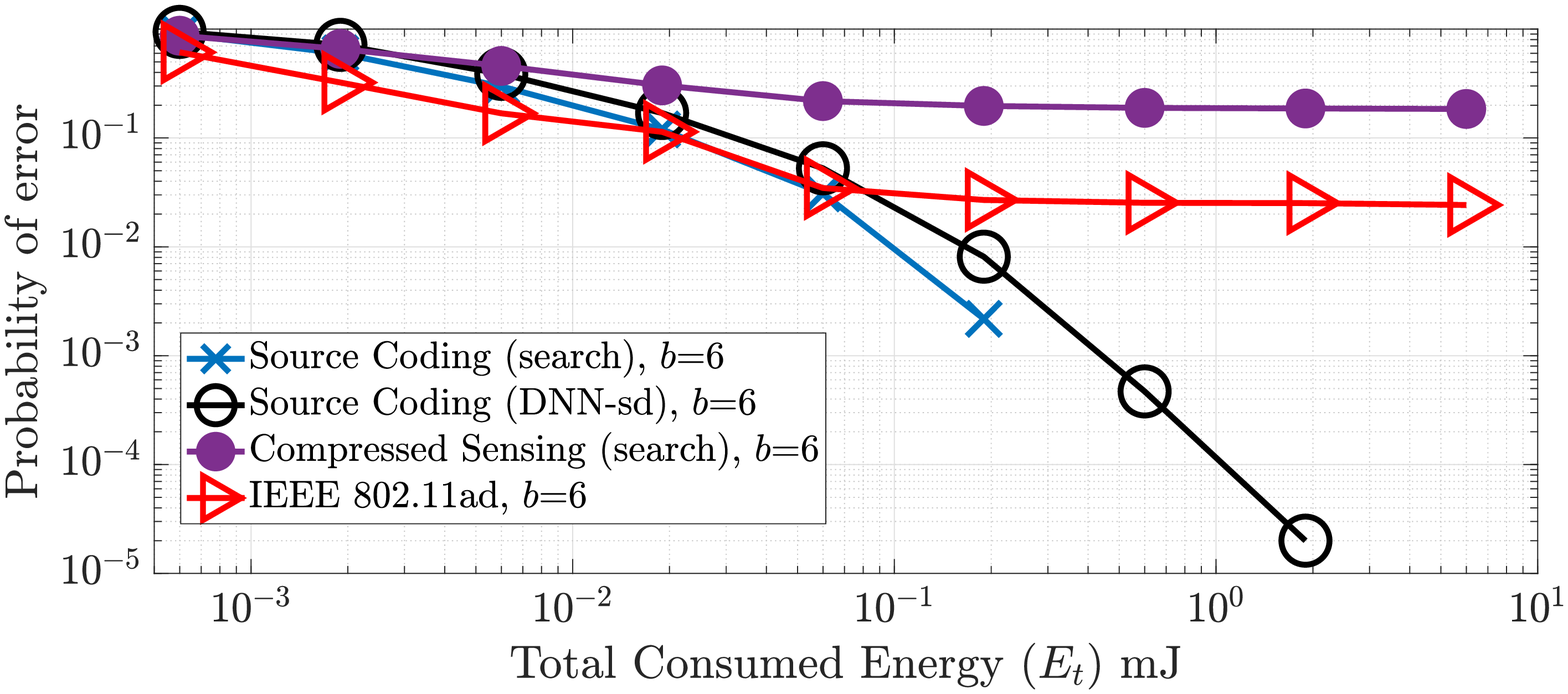}
\caption{\footnotesize Probability of Path Misdetection ($1 {-} \mathbb{P}(k \geq 1)$).}
\label{fig:compareToIEEE802.11ad_kge1} 
\end{subfigure}%
\begin{subfigure}[t]{0.33\textwidth}
\centering
\includegraphics[width=0.95\linewidth]{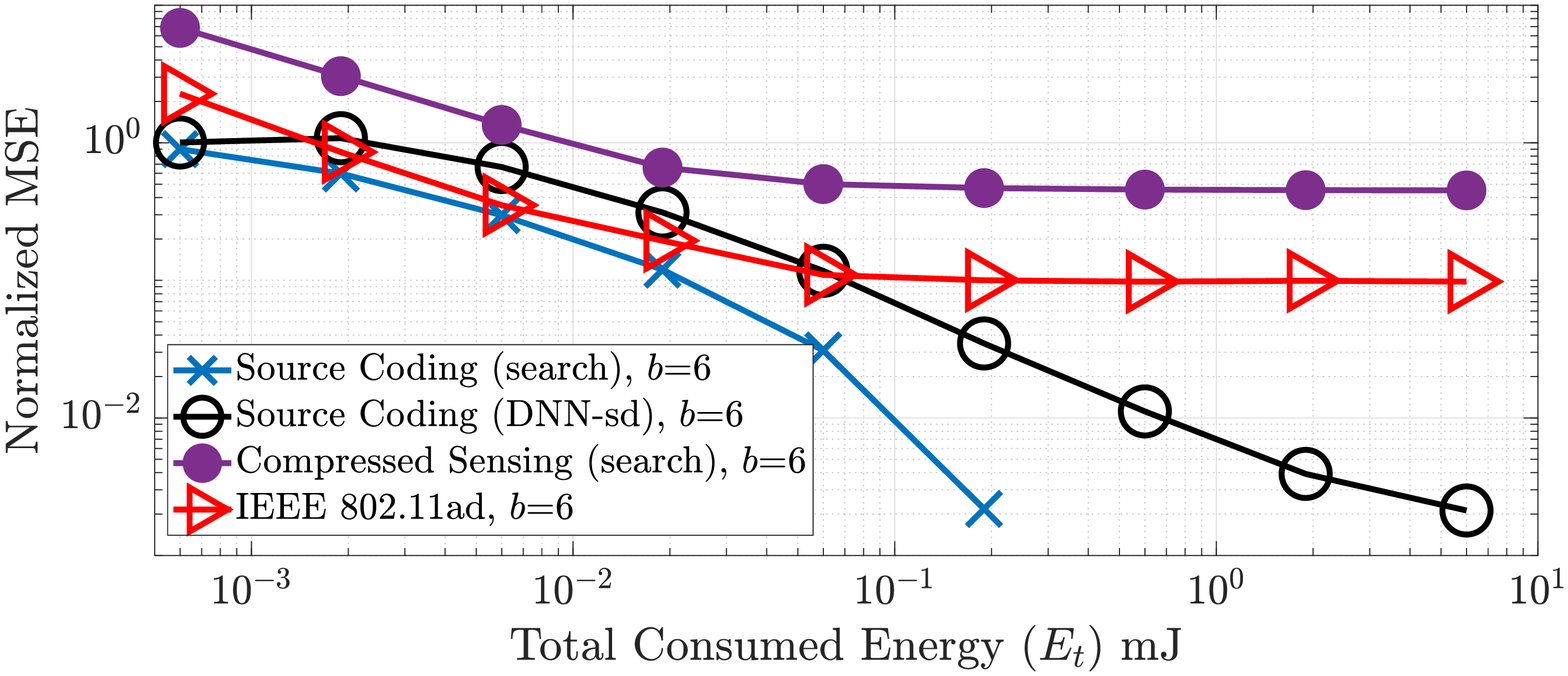}
\caption{\footnotesize Normalized MSE.}
\label{fig:compareToIEEE802.11ad_MSE} 
\end{subfigure}%
\begin{subfigure}[t]{0.33\textwidth}
\centering
\includegraphics[width=0.95\linewidth]{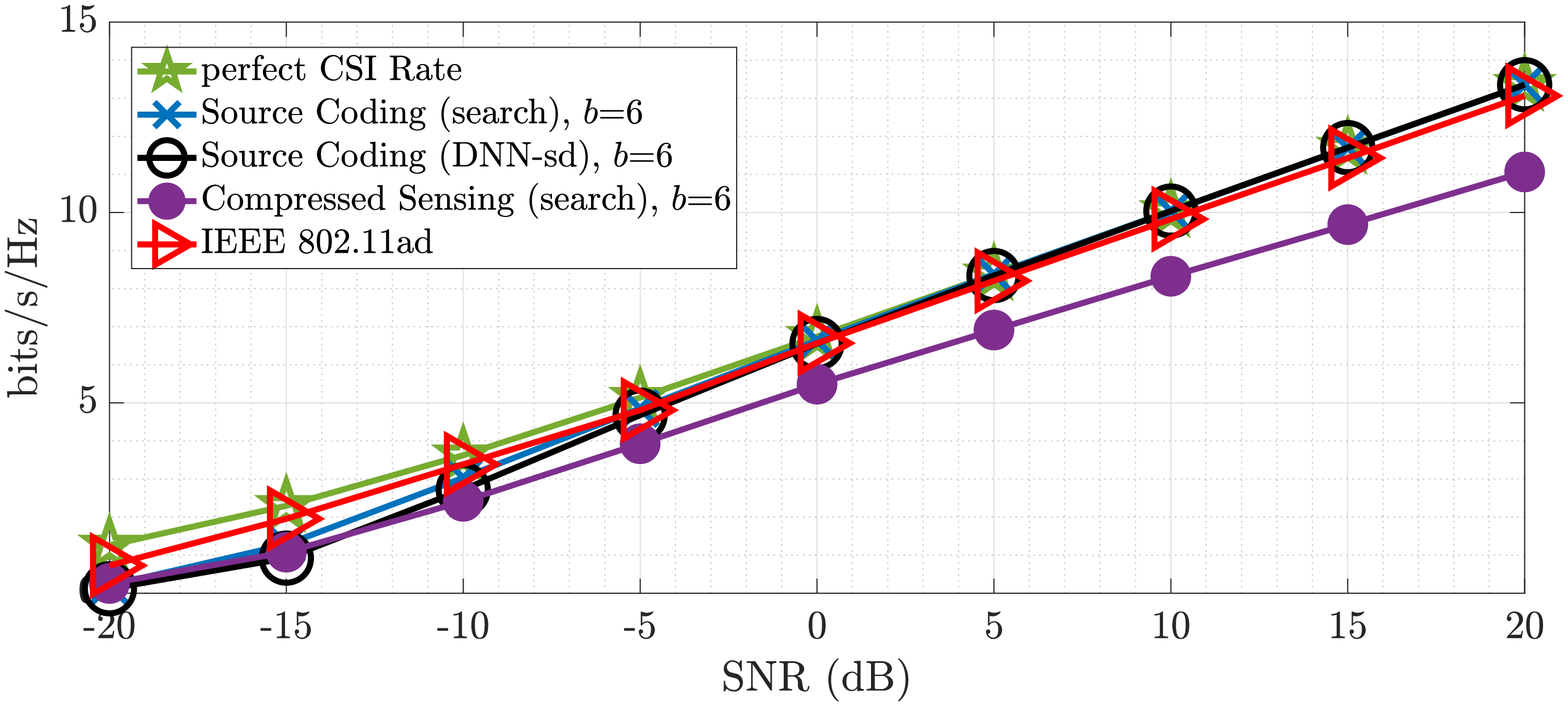}
\caption{\footnotesize Outage Rate ($R_{\text{out}}$).}
\label{fig:compareToIEEE802.11ad_Rout}
\end{subfigure}
\centering
\caption{\small Performance of single-path $15{\times}31$ channels.}
\label{fig:15x31_b=6}
\end{figure*}

\begin{figure*}[t]
\centering
\begin{subfigure}{.33\textwidth}
  \centering
  \captionsetup{justification=centering}
\includegraphics[width=0.95\linewidth]{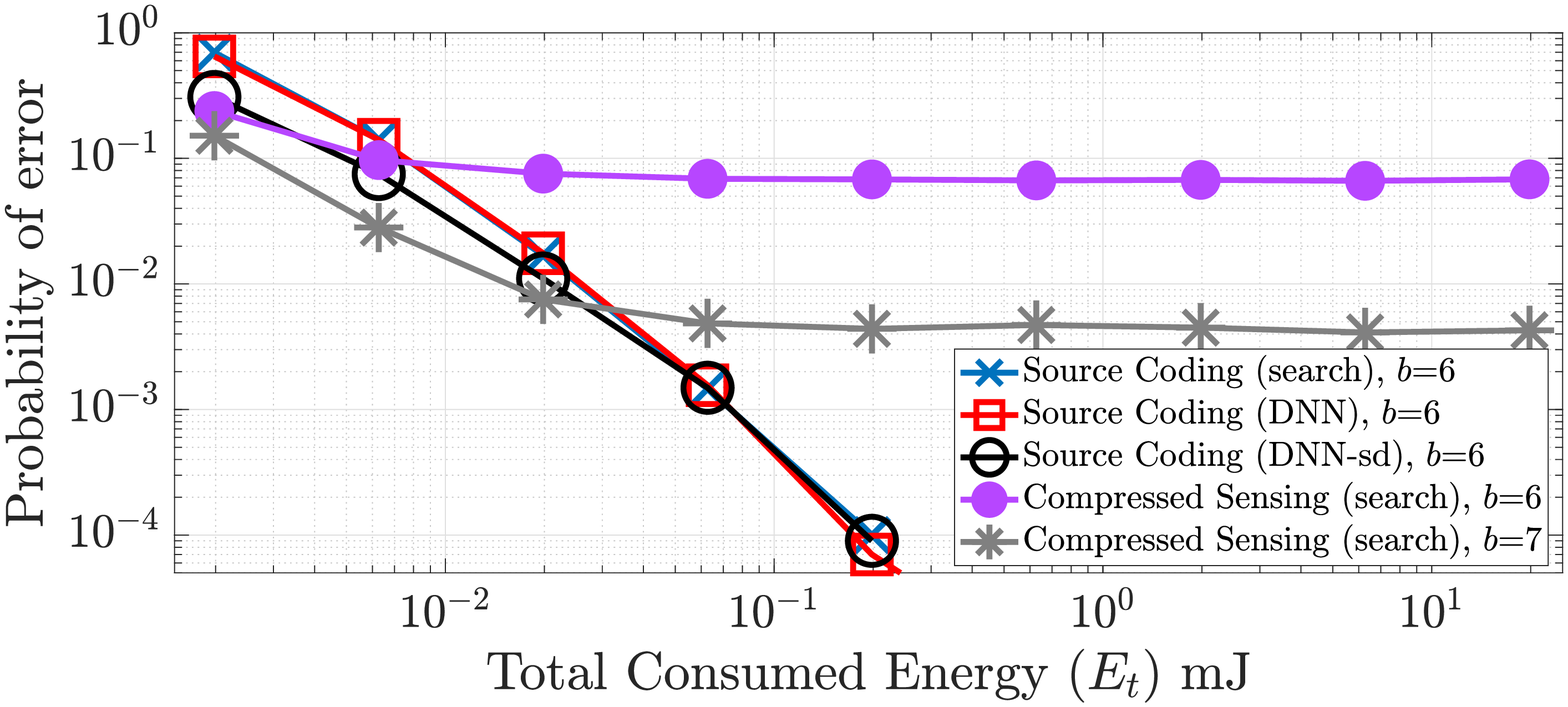}
\vspace{-0.2cm}
\caption{$1 - \mathbb{P}(k \geq 1)$}
\label{fig:23x23_b=6_k_ge1}
\end{subfigure}%
\begin{subfigure}{.33\textwidth}
  \centering
  \captionsetup{justification=centering}
\includegraphics[width=0.95\linewidth]{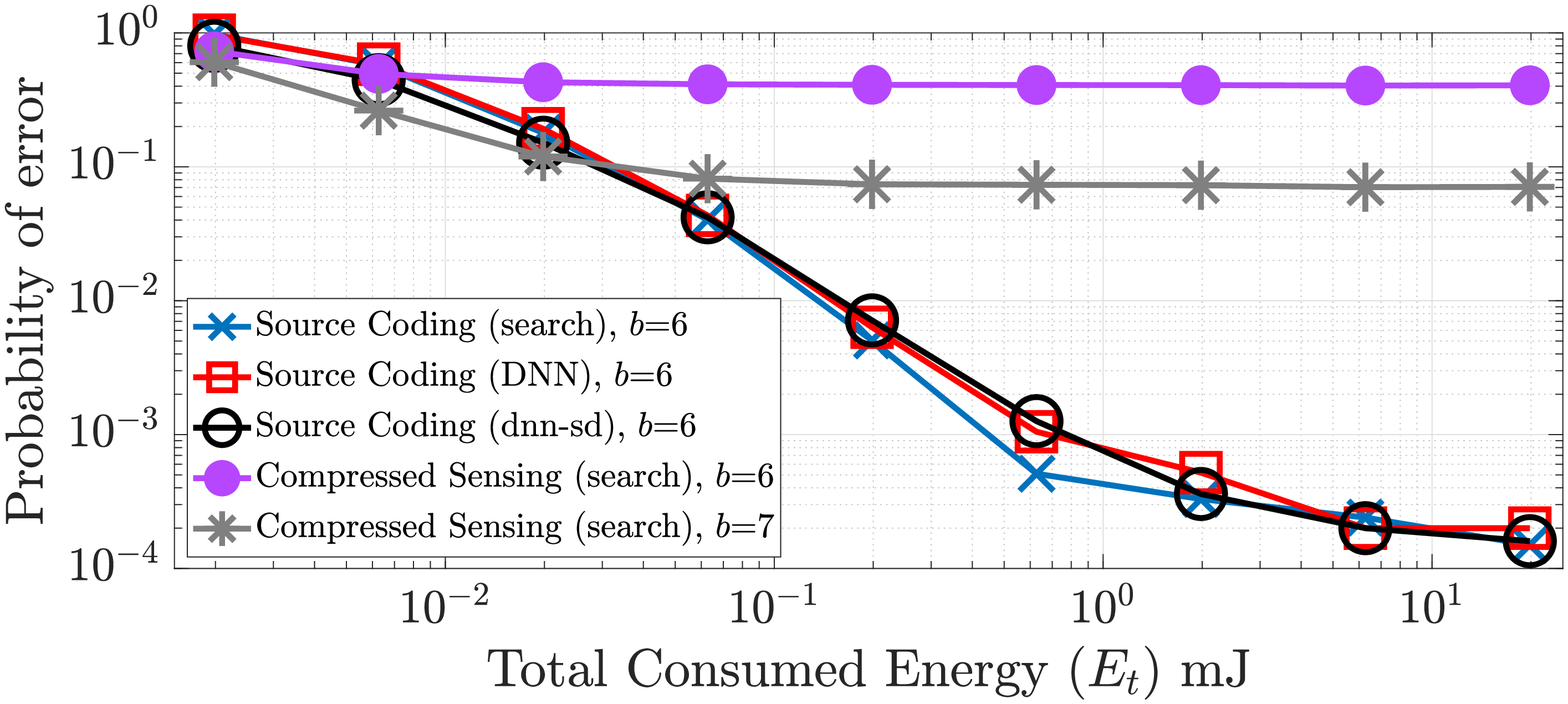}
\vspace{-0.2cm}
\caption{$1 - \mathbb{P}(k \geq 2)$}
\label{fig:23x23_b=6_k_ge2}
\end{subfigure}%
\begin{subfigure}{.33\textwidth}
  \centering
  \captionsetup{justification=centering}
\includegraphics[width=0.95\linewidth]{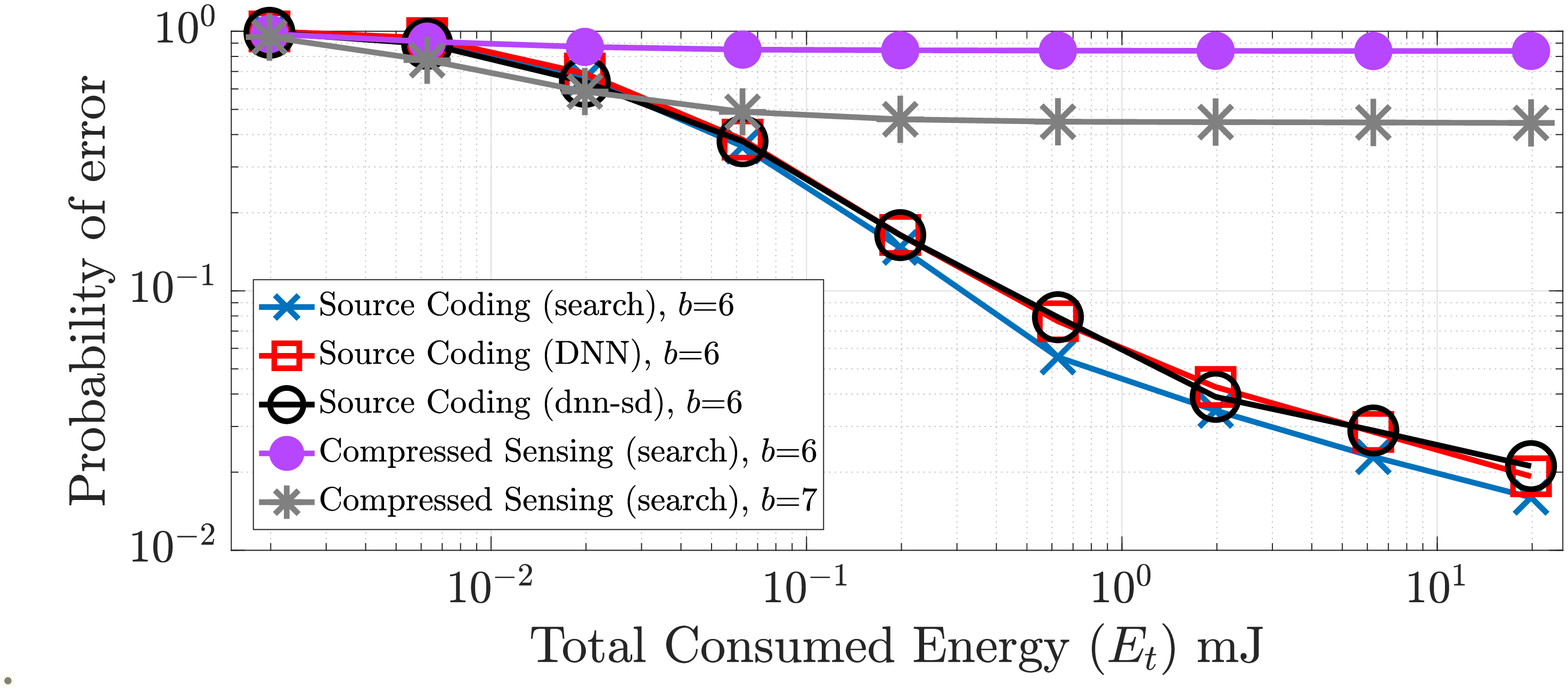}
\vspace{-0.2cm}
\caption{$1 - \mathbb{P}(k \geq 3)$}
\label{fig:23x23_b=6_k_ge3}
\end{subfigure}%
\caption{\small Beam detection probability of error for $23{\times}23$ channels with $L {\leq} 3$.}
\label{fig:23x23_b=6}
\end{figure*}
\begin{figure*}[t]
\centering
\begin{minipage}{.48\textwidth}
\includegraphics[width=0.85\linewidth]{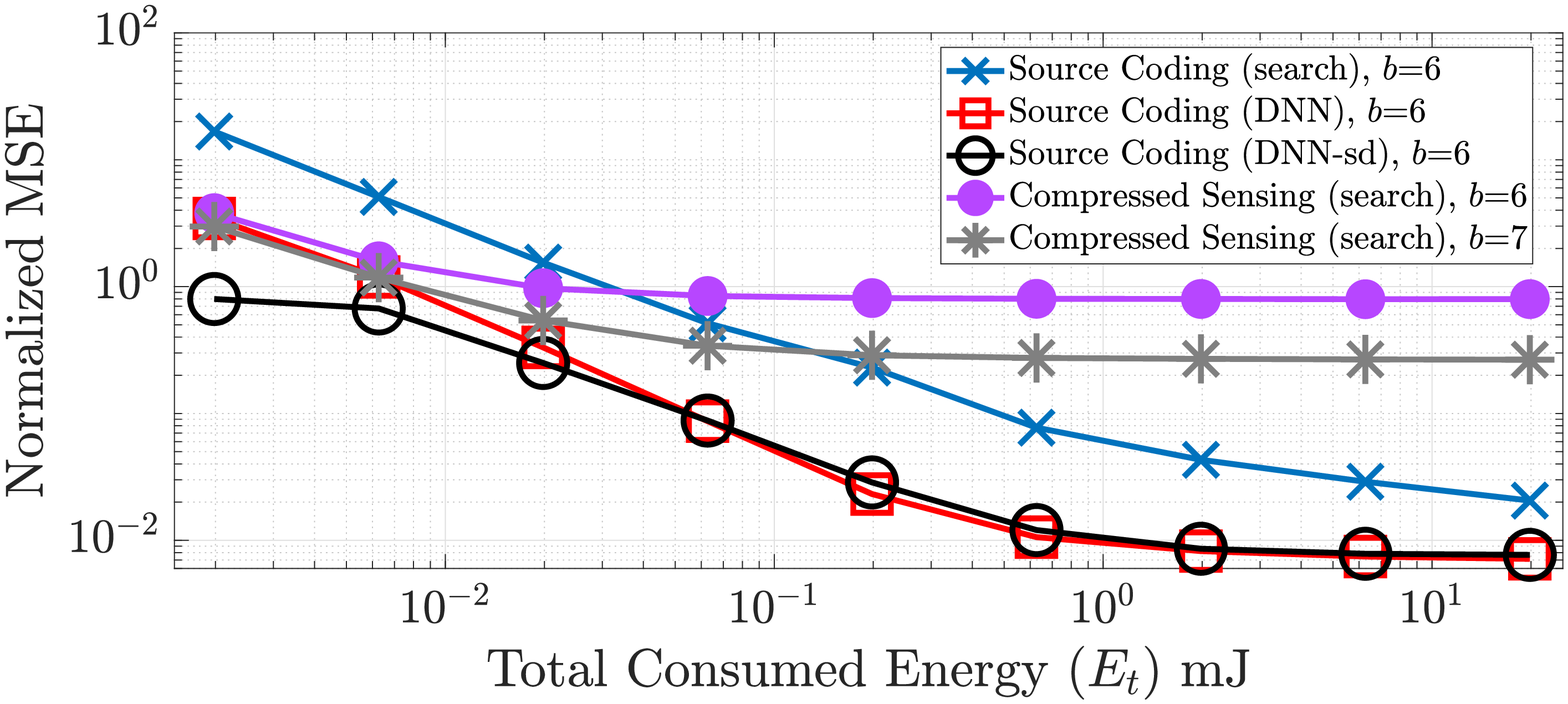}
\caption{\footnotesize Normalized MSE ($23{\times}23$ channels with $L {\leq} 3$)}
\label{fig:23x23_b=6_normalizedMSE} 
\end{minipage}%
\begin{minipage}{.48\textwidth}
\includegraphics[width=0.85\linewidth]{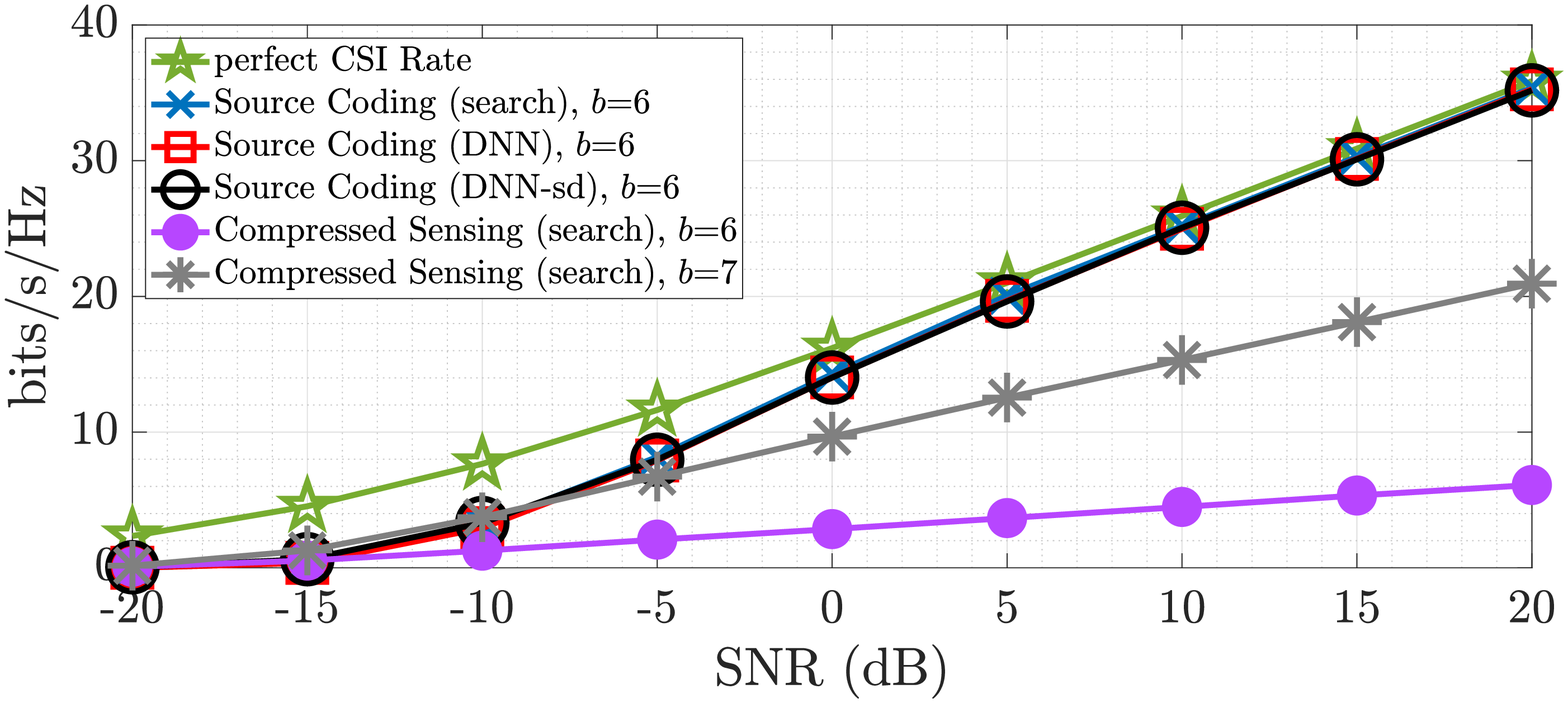}
\caption{\footnotesize Outage Rate ($R_{\text{out}}$, $23{\times}23$ channels with $L {\leq} 3$)}
\label{fig:23x23_b=6_outageRate} 
\end{minipage}
\end{figure*}

\subsection{Implemented Solutions}
\textbf{1- Source Coding:}
We test three different measurement decoding methods, which we integrate with our coding-based solution.
All three methods are used to solve each of the sub-problems depicted in Fig. (\ref{fig:multiTxRx}).
The first is the ``search'' method of \cite{shabaraBeamToN}.
The second and third decoding methods are based on DNNs, but they differ in the way they are trained, i.e., with or without measurement errors. We explain them as follows:
\begin{itemize}
\item \textbf{DNN}: Here, DNN models are trained using pure measurements, with no added noise components.
Since models are not trained with errors, only one model can be used at all SNR and ADC resolution levels.
\item \textbf{DNN-sd}: Since measurement errors tend to degrade the performance of path discovery, we try to overcome this impediment by training DNN models with error-corrupted measurements.
Since errors are dependent on ADC resolution and SNR (see Eq. (\ref{eqn:SNR})), we train multiple DNNs for different values of each. We call such model ``DNN with selective defense'' or ``DNN-sd''.
Note that the DNN-sd model is not dependent on specific path gain values since our SNR definition does not include the effect of individual path gains $\alpha_l$ and because training is done using a wide range of uniformly distributed gain values.
\end{itemize}

The DNN model parameters, including the number of layers, the number of nodes (neurons) per layer and the activation function, are selected using cross-validation.
We also select the DNN model's size such that we have a reasonably good input-output mapping performance while keeping the processing speed fairly fast. We used \texttt{tensorflow} \cite{tensorflow} for creating and using DNN models.
\textbf{\textit{Both types of DNN models are trained offline and stored in memory.}}

\textbf{2- Compressed Sensing:}
We use a similar formulation for the mmWave channel estimation problem as in \cite{Alkhateeb_2015_HowManyMeasurements,lu2019comparison}.
The tx-precoders and rx-combiners are obtained using random, uniformly distributed phase shift values. That is, the components of all $\boldsymbol{f_j}$'s and $\boldsymbol{w_i}$'s are of the form $\exp(j\vartheta)$ where $\vartheta \sim [0, 2\pi)$. Fig. \ref{fig:rx_CS} shows antenna pattern examples for random beamforming with $15$ and $31$ antennas.
For measurement decoding, we use the ``search'' method, which is the optimal $\ell_0$-norm minimization solution \cite{CS4Wireless_TipsAndTricks} for solving each of the sub-problems of channel decoding.
While this may still be too computationally complex to be of practical use, it provides an upper bound on the performance of other sparse recovery algorithms like OMP, $\ell_1$ and $\ell_2$-norm minimization, etc.

\textbf{3- 802.11ad:}
We only consider the Sector Level Sweep stage of the channel estimation scheme of 802.11ad. At this stage, the TX starts by sequentially transmitting packets in all possible transmit AoDs (sectors) while the receiver performs quasi omni-directional reception. Then, the TX and RX switch modes where TX forms a quasi omni-directional pattern while RX sweeps through all possible receive AoAs (sectors). 
The directions that reveal the highest received signal strength is denoted as the AoA and AoD of the strongest channel path.
\begin{figure*}[t]
\centering
\begin{subfigure}{.33\textwidth}
  \centering
  \captionsetup{justification=centering}
\includegraphics[width=0.95\linewidth]{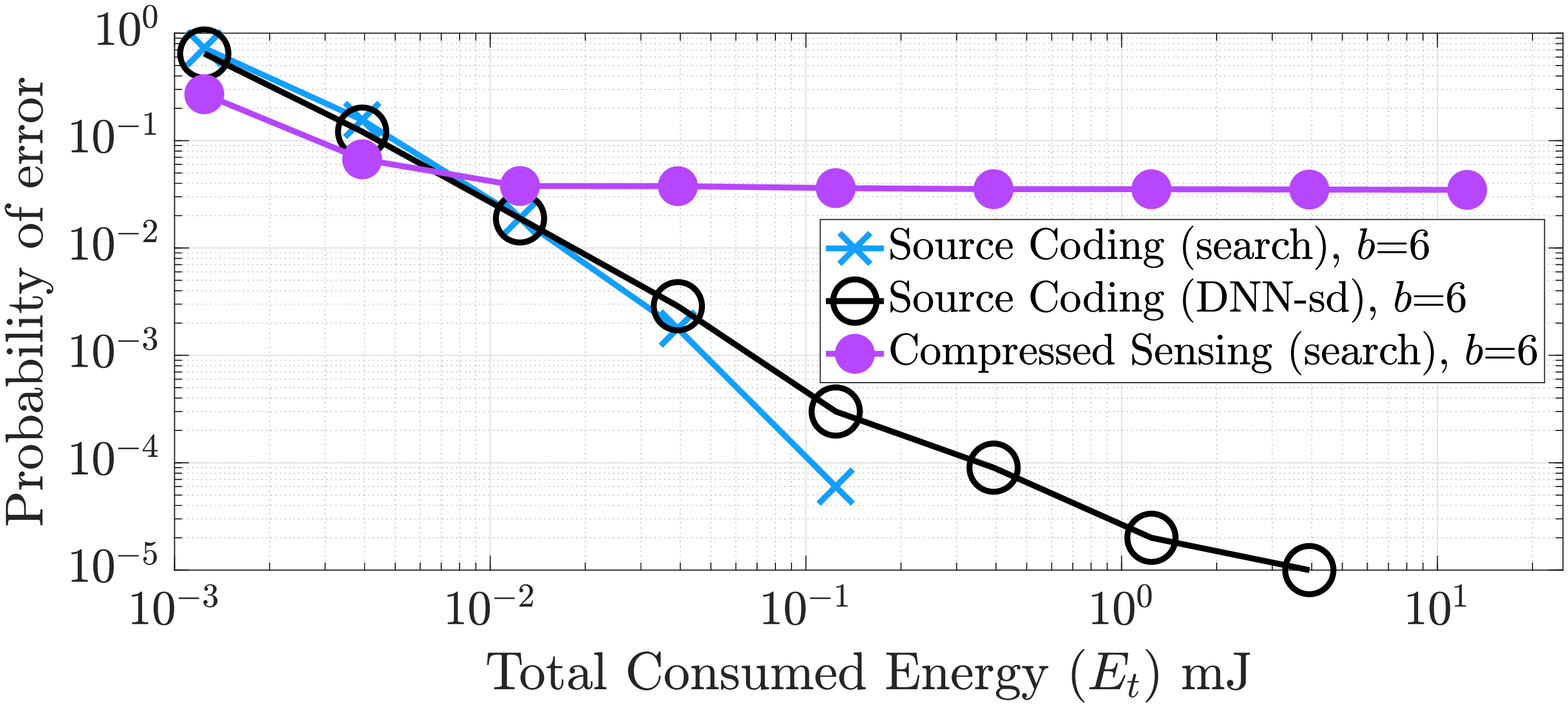}
\vspace{-0.2cm}
\caption{$1 - \mathbb{P}(k \geq 1)$}
\label{fig:15x32_b=6_k_ge1}
\end{subfigure}%
\begin{subfigure}{.33\textwidth}
  \centering
  \captionsetup{justification=centering}
\includegraphics[width=0.95\linewidth]{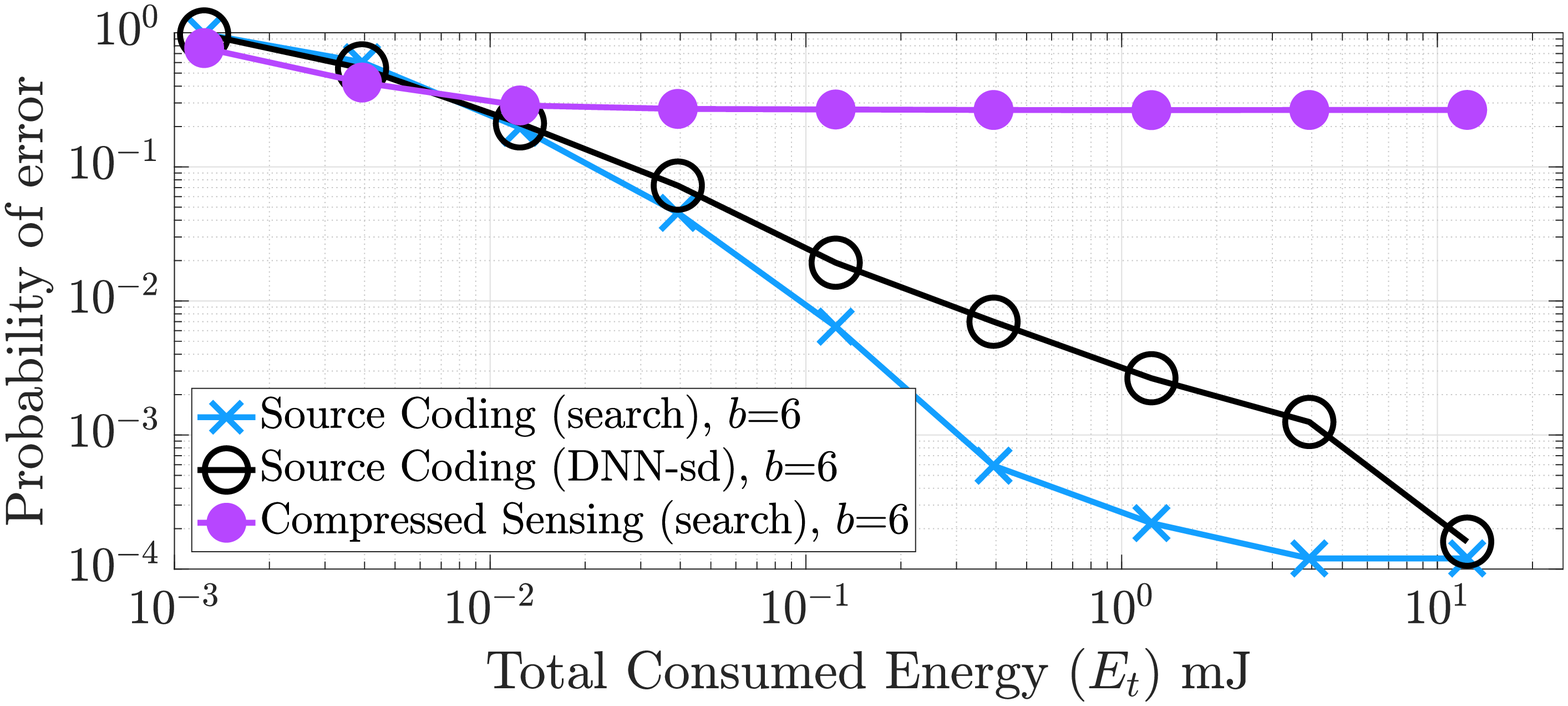}
\vspace{-0.2cm}
\caption{$1 - \mathbb{P}(k \geq 2)$}
\label{fig:15x32_b=6_k_ge2}
\end{subfigure}%
\begin{subfigure}{.33\textwidth}
  \centering
  \captionsetup{justification=centering}
\includegraphics[width=0.95\linewidth]{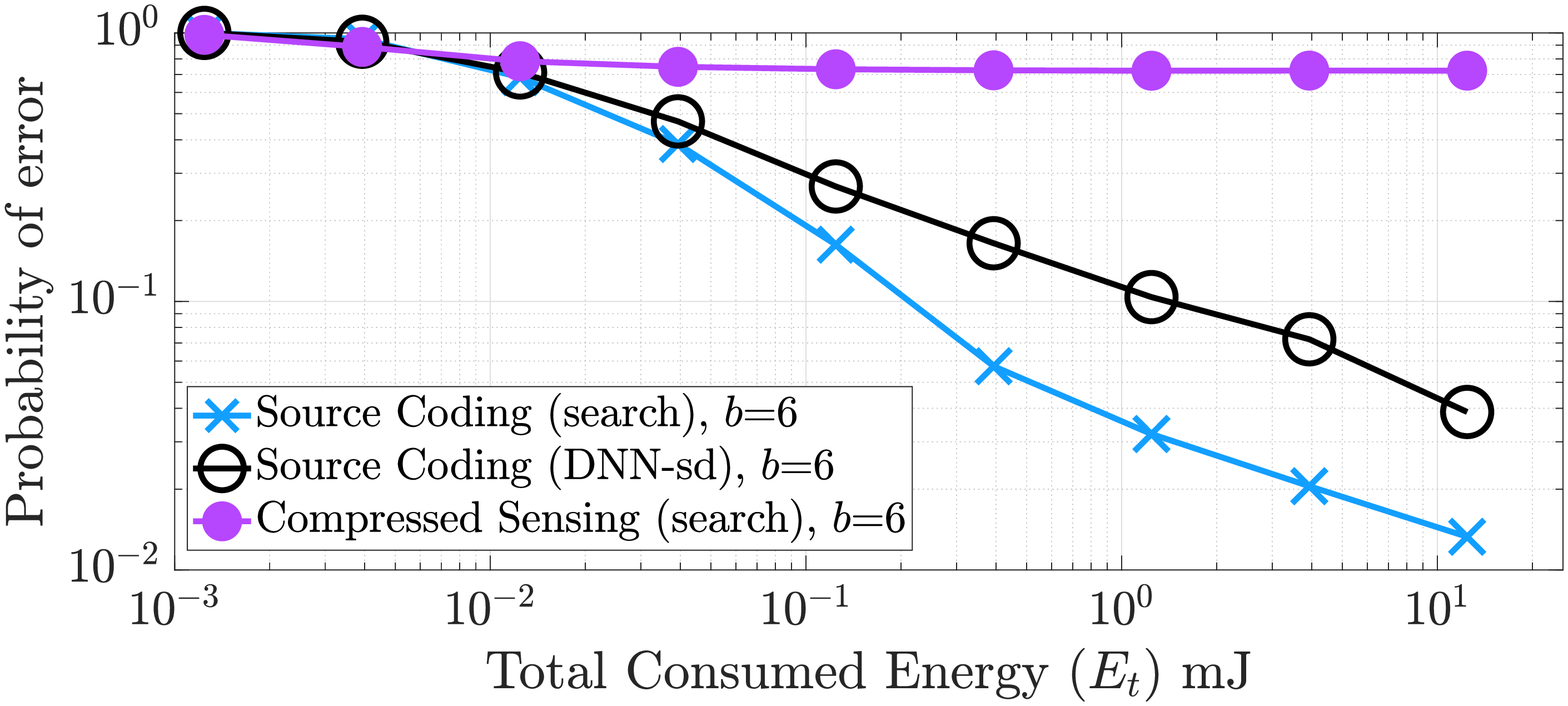}
\vspace{-0.2cm}
\caption{$1 - \mathbb{P}(k \geq 3)$}
\label{fig:15x32_b=6_k_ge3}
\end{subfigure}%
\caption{\small Beam detection probability of error for $15{\times}32$ channels with $L {\leq} 3$.}
\label{fig:15x32_b=6}
\end{figure*}
\begin{figure*}[t]
\centering
\begin{minipage}{.48\textwidth}
\includegraphics[width=0.85\linewidth]{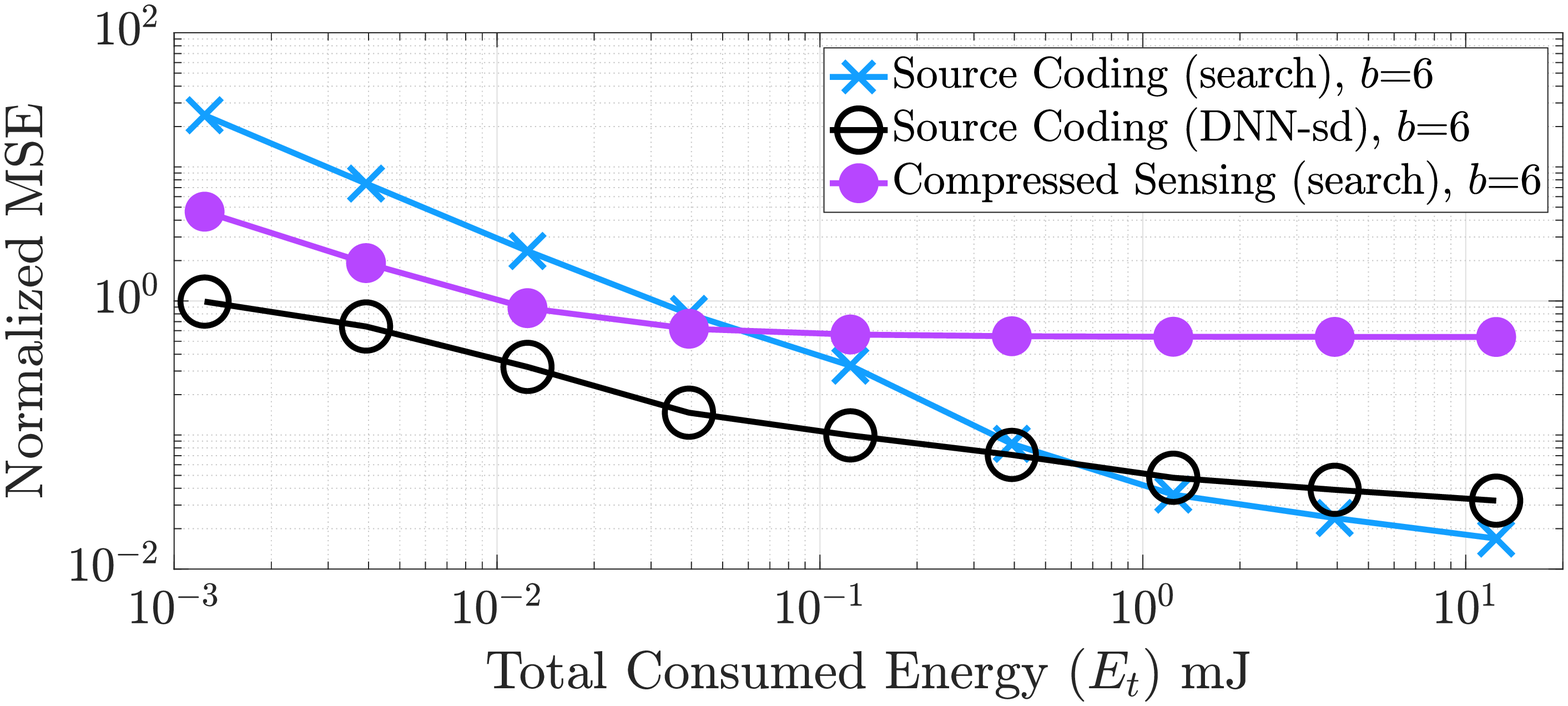}
\caption{\footnotesize Normalized MSE ($15{\times}32$ channels with $L {\leq} 3$)}
\label{fig:15x32_b=6_normalizedMSE} 
\end{minipage}%
\begin{minipage}{.48\textwidth}
\includegraphics[width=0.85\linewidth]{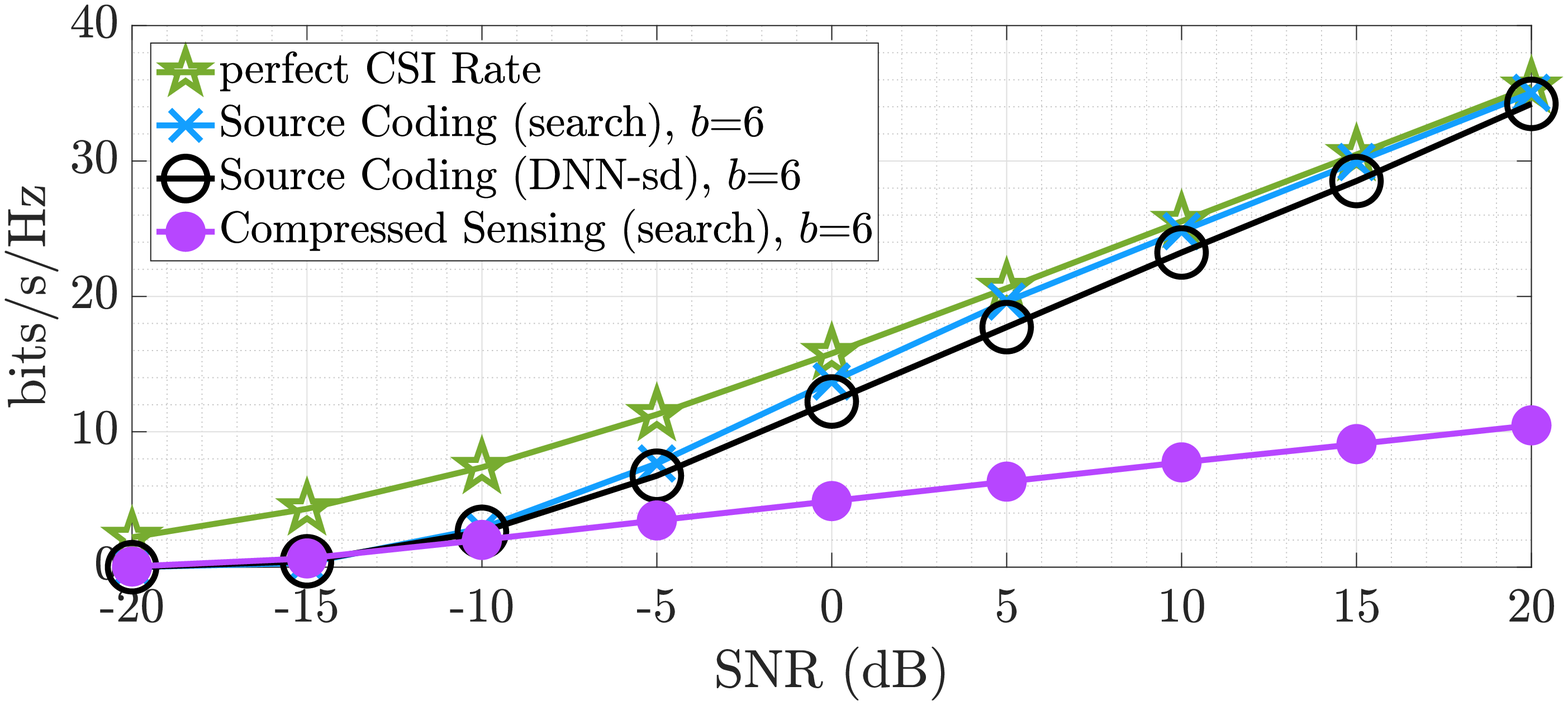}
\caption{\footnotesize Outage Rate ($R_{\text{out}}$, $15{\times}32$ channels with $L {\leq} 3$)}
\label{fig:15x32_b=6_outageRate} 
\end{minipage}
\end{figure*}

\subsection{Equating Energy Consumption}
Various channel estimation solutions may require different number of measurements and may have different beamforming gains. Thus, it would not be fair to compare them at fixed transmission power.
Instead, it is more fair to fix the total energy consumption for the whole channel estimation process of each solution.
Thus, for comparison purposes, we opt to adjust the transmit power of each scheme such that \textbf{\textit{the total amount of energy consumption for the entire measurement process remains the same}}.

\textbf{Energy Calculation:}
The energy consumption, denoted by $E_T$, is given by $E_T = m \times P_T \times \tau$, where $m$ is the number of measurements, $P_T$ is the total transmit power per measurement and $\tau$ is the time duration of one measurement.
Since the antenna patterns at TX/RX of our proposed scheme consist of multiple overlapped beams (recall Fig. \ref{fig:beamforming_demo}), the total power $P_T$ is an integer multiple of the transmit power per direction/beam $P$, which depends on the number of overlapped beams at TX and RX.
Let $o_t$ and $o_r$ denote the number of overlapped beams at the TX and RX, respectively.
Hence, we have that $P_T = o_t {\times} o_r {\times} P$.
We can further write the transmit power per beam as
$P = \text{SNR} \frac{N_0}{\mu}$ (recall Eq. (\ref{eqn:SNR})).
This gives us a total energy consumption (in millijoules (mJ)) for the entire measurement process as: $E_T = m {\times} o_t {\times} o_r {\times} \text{SNR} \frac{N_0}{\mu}{\times} \tau$.
Let $\mu {=} -88$dB and $N_0 {=} 88$dBm\footnote{
To find $\mu$ and $N_0$, we assume a channel operating at a carrier frequency $f_c {=} 60$ GHz with bandwidth $B {=} 100$ MHz and distance between TX/RX of $d{=}10$m.
Further, we assume a receiver system with noise figure NF ${=} 6$ dB and temperature $T_0 {=} 293^{\circ}$ kelvin.
The path loss constitutes both the free space path loss (FSPL) and atmospheric absorption. FSPL is given as:
$\text{FSPL} {=} -10 \log_{10}\left( \frac{4 \pi}{c} df_c \right)^{n_p}$, where $n_p {=} 2$ is the path loss exponent. Atmospheric absorption, however, can be ignored for small distances ($\leq 50$m) \cite{tomkins201560}. Hence, $\mu {=} \text{FSPL} {=} -88$dB.
The noise power (in dBm) can be given as $N_0 {=} 10 \log_{10}\left( k_B T_0 B {\times} 1000\right) {+} \text{NF}$, where $k_B$ is the Boltzmann constant.}. Finally, let $\tau \approx 23\mu$s (from IEEE 802.11ad).

\subsection{Results}

\textbf{$15 {\times} 31$ single-path channels:}
For this scenario, we choose ADCs of resolution $b{=}6$ bits. We provide results for our coding-based solution with both search and DNN-sd decoding. We also provide results for compressed sensing with search-based decoding, and IEEE 802.11ad beam alignment. We plot the results against the consumed energy $E_T$.
\textbf{Source code selection:}
We choose codes which satisfy the requirements in Theorem \ref{thm:uniquenessOfMeasurements} as follows:
At the TX side, we choose the $(31,26)$ Hamming code to design tx-precoders, while at the RX side we choose the $(15,11)$ Hamming code to design rx-combiners.
Both of which operate as syndrome source codes with generator matrices $\boldsymbol{G_t}$ and $\boldsymbol{G_r}$ of sizes $5{\times}31$ and $4{\times}15$, respectively.
Hence, we have $m_t {=} 5$ and $m_r {=} 4$, which gives us a \textbf{total number of $20$ measurements}.
For compressed sensing, we use the same number of measurements (i.e., $m=20$).
An exhaustive search, on the other hand, requires $465$ measurements (our solution represents a measurement $95\%$ reduction), while the IEEE 802.11ad scheme requires $46$ measurements ($57\%$ reduction).

For the source coding method, both the normalized MSE and probability of path misdirection results, shown in Fig. \ref{fig:15x31_b=6}, depict that DNN-sd decoding has a slightly worse performance compared to the search method.
This is a small sacrifice in performance that is traded for a huge advantage in processing speed.
The IEEE 802.11ad's method shows superior performance at low $E_T$ (below $0.1$ mJ). As $E_T$ increases, however, its performance seizes to improve, while our source coding solution keeps approaching perfect channel discovery.
When examining the outage rate, in Fig. \ref{fig:compareToIEEE802.11ad_Rout}, we see that the relatively high MSE error and probability of path misdetection of 802.11ad, does not result in a significant degradation in $R_{\text{out}}$. In fact, it has very close value to the perfect CSI capacity.
Recall that 802.11ad requires almost twice the number of channel measurements.
The compressed sensing method, on the other hand, has the lowest $R_{\text{out}}$ and the highest MSE and probability of path misdetection among all other schemes.

\begin{figure*}[t]
\centering
\begin{subfigure}{.32\textwidth}
  \centering
  \includegraphics[width=0.95\linewidth]{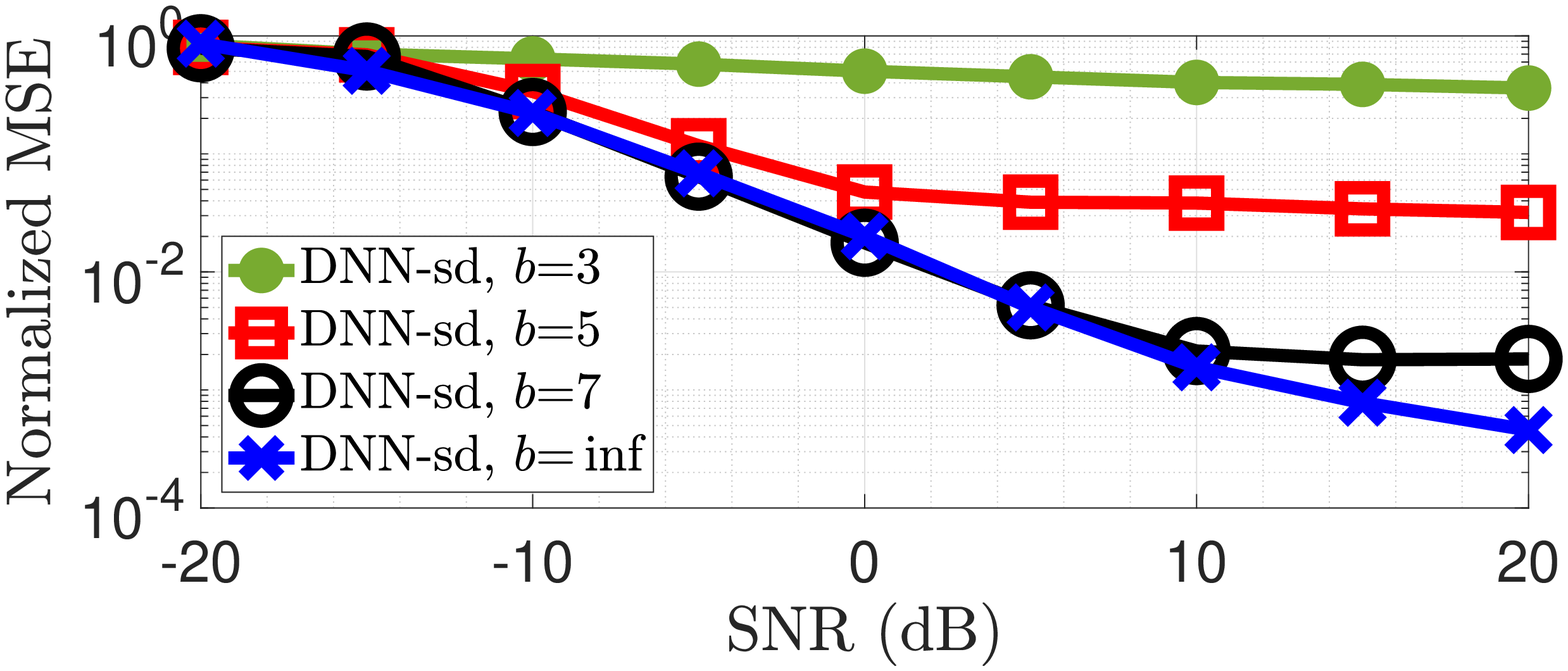}
  \vspace{-0.2cm}
  \caption{\footnotesize Normalized MSE.}
  \label{fig:allL_mse} 
\end{subfigure}%
\begin{subfigure}{.32\textwidth}
  \centering
  \includegraphics[width=0.95\linewidth]{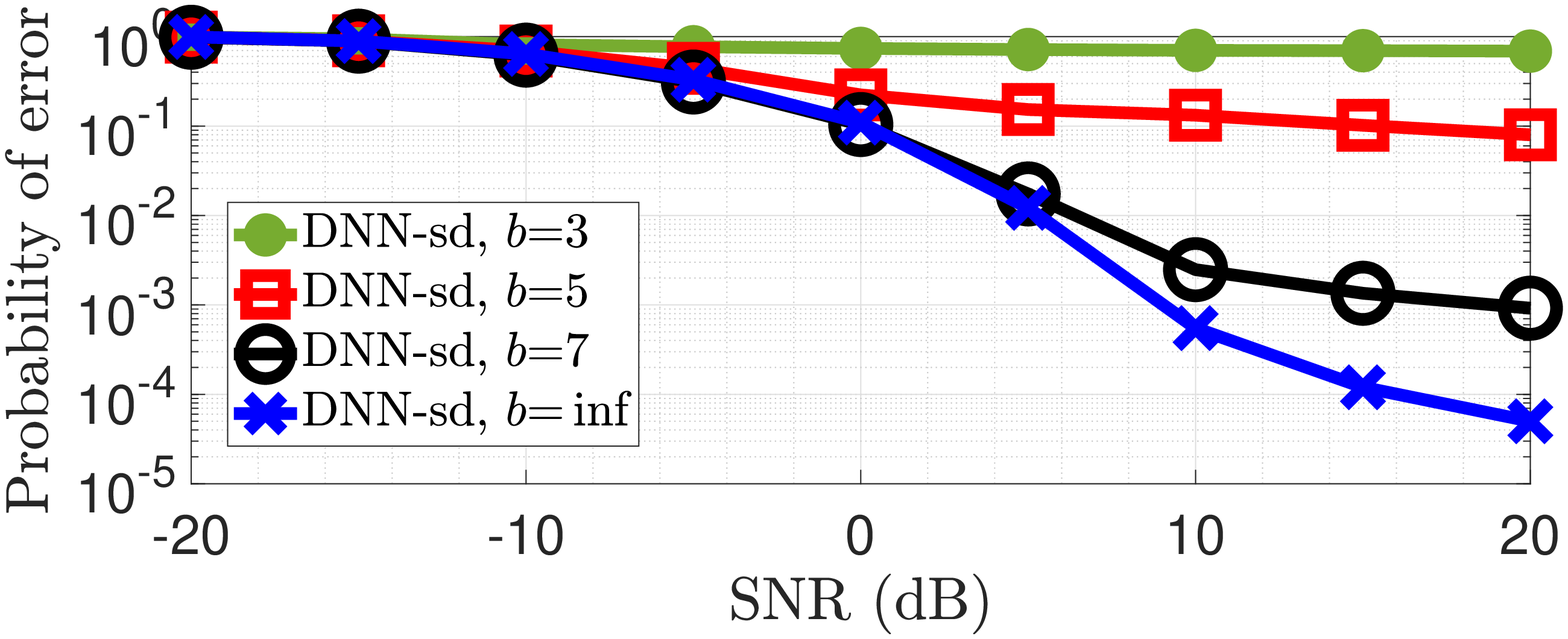}
  \vspace{-0.2cm}
  \caption{\footnotesize $1 {-} \mathbb{P}\left(k \geq 3 \right)$.}
  \label{fig:allL_kge1} 
\end{subfigure}%
\begin{subfigure}{.32\textwidth}
  \centering
  \includegraphics[width=0.95\linewidth]{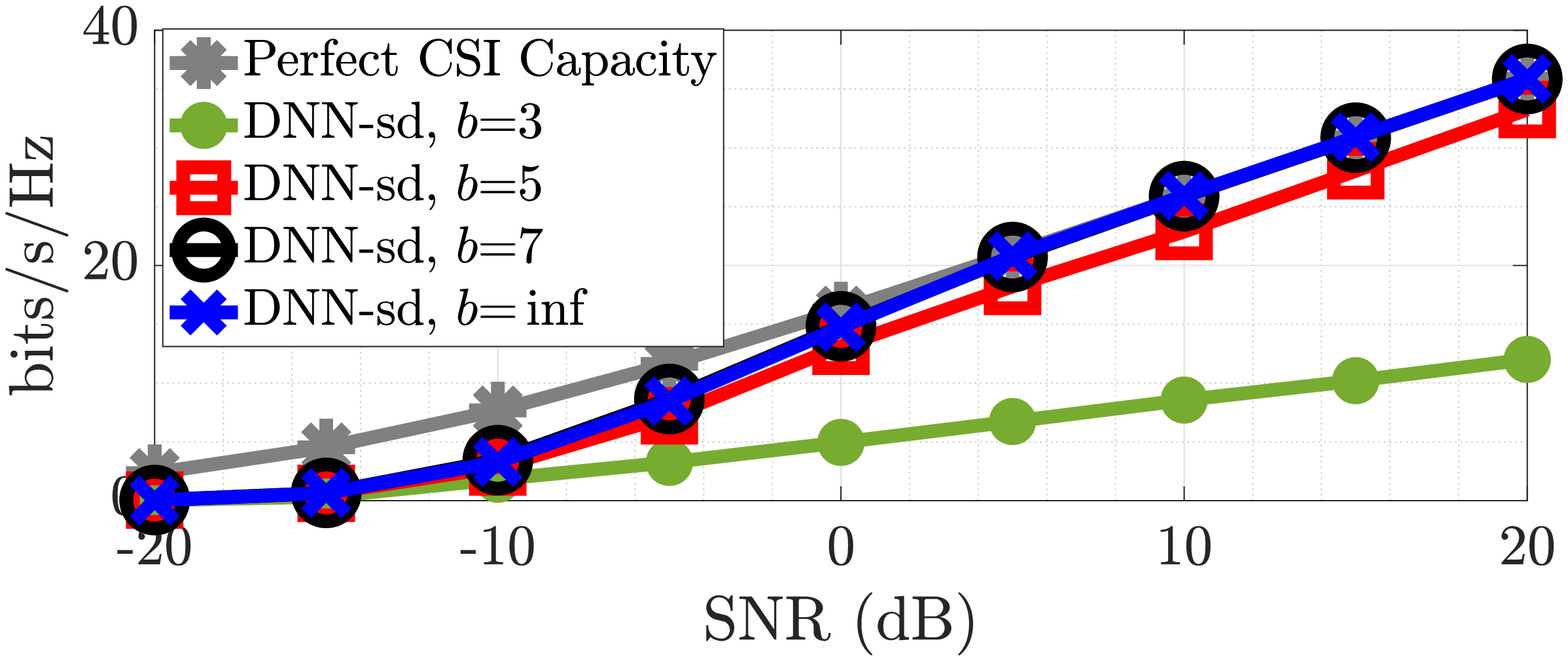}
  \vspace{-0.2cm}
  \caption{\footnotesize Outage Rate.}
  \label{fig:allL_Rout} 
\end{subfigure}
\caption{Performance of various ADC quantizations.}
\end{figure*}

\textbf{$23 \times 23$ multi-path channels:}
This is a more challenging multi-path scenario where, in addition to the previous solutions, we also investigate the performance of DNN decoding for which training is done with pure uncorrupted measurements.
\textbf{Source code selection:}
For this channel, since $n_r = n_t$, the same source code works for designing both tx-precoders and rx-combiners.
The \textbf{\textit{perfect binary Golay code}} used as a syndrome source code is a suitable choice for this problem.
It has a generator matrix of size $11{\times}22$, hence, we have $m_t {=} m_r {=} 11$, and the \textbf{total number of required channel measurements is $m_t{\times}m_r {=} 121$}.
Compared to the exhaustive search approach, which requires scanning all $529$ combinations of TX/RX angular directions, this represents $\boldsymbol{75\%}$ measurement reduction.
For compressed sensing, we use the same number of precoders and combiners, as well.

First, for the probability of path detection, shown in Fig. \ref{fig:23x23_b=6}, we notice very close performance for all three measurement decoding methods (search, DNN and DNN-sd) when integrated with our source coding solution.
This suggests that DNNs are very efficient. And, while DNN-sd has a slight edge over DNN, the improvement is not significant.
Hence, it is possible to user fewer DNNs trained over larger ranges of error components.
Interestingly, however, in Fig. \ref{fig:23x23_b=6_normalizedMSE}, we observe that the search decoding of our source-coding-based measurements tend to have significantly larger MSE. This indicates that DNNs tend to suppress the error components in the estimated channel values, even though the correct channel paths may not be efficiently discovered.
This is an artifact of DNN training which suggests that we may be able to improve the DNN's decoding performance if they were directly trained to discover the channel paths rather than just decrease the MSE of the channel estimates.

Compressed sensing, however has significantly worse performance in terms of path discover, MSE and outage Rate.
At very low $E_T$, CS's performance improves as $E_T$ gets higher, but at $E_T \geq 0.7$ the improvement stops.
This suggests that the CS-based solution is more sensitive to quantification error.
This is verified by using a higher resolution ADC ($b {=} 7$ bits) which shows a significant improvement of performance.

\textbf{$15 \times 32$ multi-path channel:}
The performance results under this scenario is very similar to the $23{\times}23$ channel setting shown above.
Specifically, Fig. \ref{fig:15x32_b=6} shows the probability of path discovery where our coding method with search decoding shows superior performance compared to DNN-sd decoding. On the contrary, search has worse MSE compared to DNN-sd (as shown in Fig.\ref{fig:15x32_b=6_normalizedMSE}). Compressed sensing, has close performance to our proposed solution at low $E_T$. At high $E_T$, however, its performance only sees marginal improvement, unlike our proposed solution which keeps approaching perfect channel discovery.
\textbf{Code selection:}
We choose codes whose $m_r = 11$ and $m_t = 16$. The \textbf{total number of required channel measurements is $176$}, which constitutes $63.3\%$ measurement reduction compared to exhaustive search.

\subsection{Effect of ADC resolution on performance:}
Now, we provide and compare results for $23 {\times} 23$ channels with ADCs of $b {=} 3,5$, and $7-$bit resolution, inas well as the ideal $b = {\infty}$.
We only show results for DNN-sd decoding since the performance of the other two methods compare similarly to the trends of the ideal ADC case shown above.
In Fig. \ref{fig:allL_mse}, we plot the MSE. As expected, we see that MSE is inversely proportional to resolution. We can also see that as SNR increases, higher resolution is required to keep the MSE close to that of ideal ADCs.
For instance, $b {=} 5$ is reasonably good up to SNR $= -5$dB, while $b {=} 7$ is very close to $b{=}\infty$ up to SNR $= 5$dB. Even at high SNR, the $b{=}7$ curve has a gap with ideal ADCs that is smaller than $5\times10^{-3}$.
Similar performance trend is exhibited for the probability of path discovery, shown in Fig. \ref{fig:allL_kge1}.
Finally, the outage rate is depicted in Fig. \ref{fig:allL_Rout}. We see that $b{=}7$ results in $R_{\text{out}}$ that is almost identical to that of the ideal ADC, and that both of which are very close to the perfect CSI capacity. We also see that at low SNR, there is a considerable gap between the perfect CSI capacity and outage rate even for ideal ADCs.

\section{Conclusion}

In this work, we treat the mmWave channel estimation problem a source compression problem.
Our goal is to reconstruct the channel matrix using a small number of measurements.
We exploit linear binary source codes for encoding the channel (do measurements) and a deep neural network based algorithm for measurement decoding (channel reconstruction).
We are able, using a small number of measurements, to obtain high quality channel estimates.
The lower bound on the achievable number of measurements is accurately characterized.
Through simulation, the superiority of our proposed solution is demonstrated, in comparison to compressed-sensing-based solutions and IEEE 802.11ad's beam alignment.

\begin{appendices}

\section{Proof Of Lemma \ref{lemma:vecIndOverFFAndR}}
\begin{proof}
Consider a set of $n-$dimensional linearly independent vectors, $\boldsymbol{v_1}, \dots, \boldsymbol{v_m}$ defined over $\mathbb{F}_2$. Then, construct a matrix $\boldsymbol{M}_{\mathbb{F}_2}$ whose columns are $\boldsymbol{v_1}, \dots, \boldsymbol{v_m}$. Since $\boldsymbol{v_i}$'s are independent, then $\boldsymbol{M}_{\mathbb{F}_2}$ has full column rank, i.e., $\boldsymbol{M}_{\mathbb{F}_2}$ is left-invertible over $\mathbb{F}_2$ ($m {\leq} n$).
Thus $\boldsymbol{M}_{\mathbb{F}_2}$ has an $m{\times}m$ minor, call it $\boldsymbol{A}_{\mathbb{F}_2}$ whose determinant is non-zero.
Now consider the matrix $\boldsymbol{M}$, defined over $\mathbb{C}$, whose elements are the $0$ and $1$ real scalars corresponding to $\mathrm{0}_{\mathbb{F}_2}$ and $\mathrm{1}_{\mathbb{F}_2}$ values of $\boldsymbol{M}_{\mathbb{F}_2}$.
Let $\boldsymbol{A}$ be its minor corresponding to $\boldsymbol{A}_{\mathbb{F}_2}$ of $\boldsymbol{M}_{\mathbb{F}_2}$.
By lemma \ref{lemma:nonZeroDets} (in Appendix \ref{append:SecondLemma}), we have $\det\left(\boldsymbol{A}\right) {\neq} 0$. Thus, $\boldsymbol{M}$ is also left-invertible, hence its columns are linearly independent.
\end{proof}
\label{append:FirstLemmaProof}

\section{Lemma \ref{lemma:nonZeroDets}}
\begin{lemma}\label{lemma:nonZeroDets}
Let $\boldsymbol{A}_{\mathbb{F}}$ and $\boldsymbol{A}$ be $n {\times} n$ matrices defined over $\mathbb{F}_2$ and $\mathbb{R}$, respectively. Let ${a_{\mathbb{F}}}_{i,j}$, the elements of $\boldsymbol{A}_{\mathbb{F}}$, be scalars in $\{\mathsf{0}_{\mathbb{F}_2} , \mathsf{1}_{\mathbb{F}_2} \}$, while $a_{i,j}$ the elements of $\boldsymbol{A}$, be scalars in $\{0 ,1\} {\subseteq} \mathbb{R}$.
Suppose that $\boldsymbol{A}_{\mathbb{F}}$ has non-zero determinant, i.e., $\det\left( \boldsymbol{A}_{\mathbb{F}} \right) {\neq} \mathsf{0}_{\mathbb{F}_2}$.
If we define $\boldsymbol{A}$ such that
$a_{i,j} = 0$ if ${a_{\mathbb{F}}}_{i,j} = \mathsf{0}_{\mathbb{F}_2}$, and $a_{i,j} = 1$ if ${a_{\mathbb{F}}}_{i,j} = \mathsf{1}_{\mathbb{F}_2}$ for all $1 {\leq} i,j {\leq} n$. Then, $\det \left( \boldsymbol{A} \right) \neq 0$.
\end{lemma}

\begin{proof}
Recall that the determinant of a square matrix defined over a commutative ring is given by the Leibniz formula \cite{grinberg2016notes}.
Since $\mathbb{F}_2$ is a finite field (with $2$ elements), it constitutes a commutative ring. Moreover, $\mathbb{R}$ is a commutative ring \cite{grinberg2016notes}. Therefore, both determinants of $\boldsymbol{A}_{\mathbb{F}}$ and $\boldsymbol{A}$ can be computed using the same exact formula.
Since, finite field arithmetic over the prime field $\mathbb{Z}_2$ is the integers $modulo$ $2$, then we can write $\det (\boldsymbol{A}_{\mathbb{F}}) = \det ( \boldsymbol{A} ) \mod 2 $.
Thus, $\exists q \in \mathbb{Z}$ (the set of integers), such that
$\det ( \boldsymbol{A} )
= q \times 2 +  \det (\boldsymbol{A}_{\mathbb{F}})
= q \times 2 +  1$,
were the latter equation follows from the fact that $\det (\boldsymbol{A}_{\mathbb{F}}) \neq \mathsf{0}_{\mathbb{F}_2} \Longleftrightarrow \det (\boldsymbol{A}_{\mathbb{F}}) {=} \mathsf{1}_{\mathbb{F}_2}$.
Therefore, $\det (\boldsymbol{A})$ \textit{\textbf{is an odd integer}}, which implies that $\det (\boldsymbol{A}) {\neq} 0$, concluding our proof.
\end{proof}
\label{append:SecondLemma}

\section{Proof of Proposition \ref{prop:minSingularValStdForm}}
\begin{proof}
We will prove that adding an extra column $\boldsymbol{p} \in \mathbb{R}^m$ to any \textbf{full rank} matrix $\boldsymbol{M}$ of size $m{\times}k$ with $m \leq k$ (i.e., $\text{rank}\left(\boldsymbol{M}\right) {=} m$) \textbf{\textit{does not reduce}} its singular values.

Let $\boldsymbol{M_p} = \begin{pmatrix}[c|c]
\boldsymbol{M} & \boldsymbol{p}
\end{pmatrix}$ be an $m{\times}k{+}1$ matrix. Then, we can obtain the singular values of $\boldsymbol{M_p}$ as the positive square roots of the eigenvalues of $\boldsymbol{M_p}\boldsymbol{M}^T_{\boldsymbol{p}}$, which can be written as:
\whencolumns
{
\small
\begin{equation}
\boldsymbol{M_p} \boldsymbol{M}^T_{\boldsymbol{p}}
=
\begin{pmatrix}[c|c]
\boldsymbol{M} & \boldsymbol{p}
\end{pmatrix}
\begin{pmatrix}[c|c]
\boldsymbol{M} & \boldsymbol{p}
\end{pmatrix}^T
=
\boldsymbol{M} \boldsymbol{M}^T + \boldsymbol{p} \boldsymbol{p}^T
\end{equation}
}
{
\begin{align}
\boldsymbol{M_p} \boldsymbol{M}^T_{\boldsymbol{p}}
&=
\begin{pmatrix}[c|c]
\boldsymbol{M} & \boldsymbol{p}
\end{pmatrix}
\begin{pmatrix}[c|c]
\boldsymbol{M} & \boldsymbol{p}
\end{pmatrix}^T \\
&=
\boldsymbol{M} \boldsymbol{M}^T + \boldsymbol{p} \boldsymbol{p}^T
\end{align}
}
Since $\boldsymbol{p} \boldsymbol{p}^T {\succeq} 0$ (i.e., positive semidefinite), then we must have
$\boldsymbol{M_p} \boldsymbol{M}^T_{\boldsymbol{p}} - \boldsymbol{M} \boldsymbol{M}^T  \succeq 0$.
Let $\sigma_i\left( \cdot \right)$ denote the $i^{\text{th}}$ largest singular value of a matrix. Then, we have $\sigma_i\left( \boldsymbol{M_p} \boldsymbol{M}^T_{\boldsymbol{p}}\right)
\geq
\sigma_i\left( \boldsymbol{M} \boldsymbol{M}^T\right) \; \forall i = 1, \dots, m$, which implies
\begin{align}
\Longrightarrow
\sigma_{\text{min}}\left( \boldsymbol{M_p} \boldsymbol{M}^T_{\boldsymbol{p}}\right)
&\geq
\sigma_{\text{min}}\left( \boldsymbol{M} \boldsymbol{M}^T\right)\\
\Longrightarrow
\sigma_{\text{min}}\left( \boldsymbol{M_p} \right)
&\geq
\sigma_{\text{min}}\left( \boldsymbol{M} \right) \label{eqn:helpingResult}
\end{align}
Define $\boldsymbol{G}^{(i)} {\triangleq} \begin{pmatrix}[c|c|c|c]
\boldsymbol{g}_1 & \boldsymbol{g}_2 & \dots & \boldsymbol{g}_i
\end{pmatrix}$, where $\boldsymbol{g}_j$ is the $j^{th}$ column of $\boldsymbol{G}$.
Then, by sequentially applying the result shown in Eq. (\ref{eqn:helpingResult})  (by adding columns of $\boldsymbol{P}$ in Eq. (\ref{eqn:minSingularValStdForm})), we obtain
\whencolumns
{
\begin{equation*}
\small
\sigma_{\text{min}}\left( \boldsymbol{G} \right) =
\sigma_{\text{min}}\left( \boldsymbol{G}^{(n)} \right) \geq\sigma_{\text{min}}\left( \boldsymbol{G}^{(n-1)} \right) \geq
\dots
\geq
\sigma_{\text{min}}\left( \boldsymbol{G}^{(m+1)} \right) \geq
\sigma_{\text{min}}\left( \boldsymbol{G}^{(m)} \right)
=\sigma_{\text{min}}\left( \boldsymbol{I}_m \right)
= 1 \qedhere
\end{equation*}
}
{
\begin{equation}
\begin{split}
\sigma_{\text{min}}\left( \boldsymbol{G} \right) =
\sigma_{\text{min}}\left( \boldsymbol{G}^{(n)} \right) \geq\sigma_{\text{min}}\left( \boldsymbol{G}^{(n-1)} \right) \geq
\dots \\
\geq
\sigma_{\text{min}}\left( \boldsymbol{G}^{(m+1)} \right) \geq
\sigma_{\text{min}}\left( \boldsymbol{G}^{(m)} \right)
=\sigma_{\text{min}}\left( \boldsymbol{I}_m \right)
= 1 \qedhere
\end{split}
\end{equation}
}
\end{proof}
\label{append:SecondPropositionProof}

\end{appendices}

\bibliographystyle{IEEEtran}
\linespread{1}
\bibliography{MyRef}

\begin{IEEEbiography}
[{\includegraphics[width=1in,height=1.25in,clip,keepaspectratio]{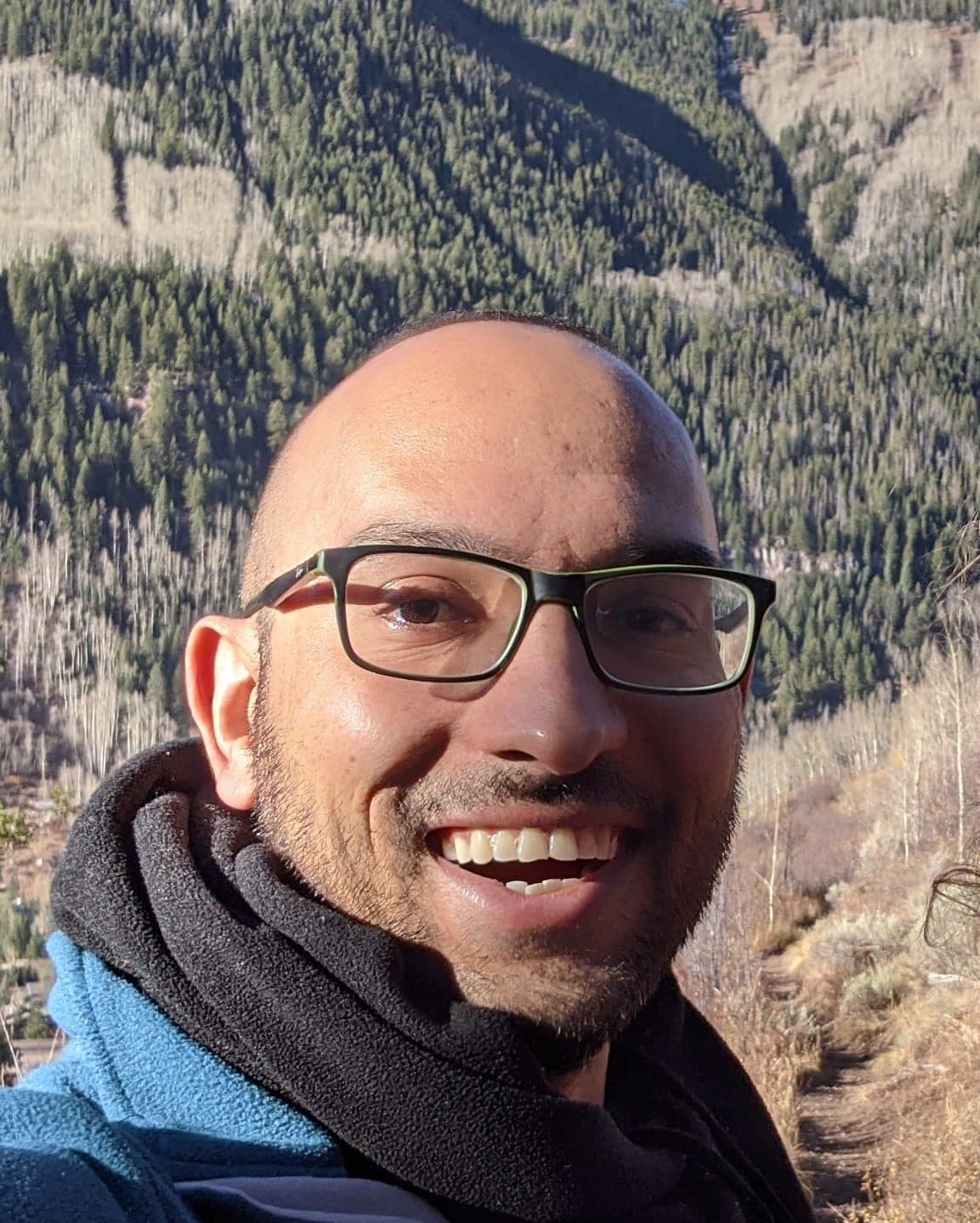}}]{Yahia Shabara}
received his B.Sc. degree in Electrical Engineering from Alexandria University, Alexandria, Egypt, in 2012 and the M.Sc. degree in wireless communications from Nile University, Giza, Egypt, in 2015. He is currently pursuing the Ph.D. degree with the Dept. of Electrical and Computer Engineering, The Ohio State University, Columbus, Ohio. His research interests include wireless communications, computer networks, machine learning, information theory and network security.
\end{IEEEbiography}

\begin{IEEEbiography}[{\includegraphics[width=1in,height=1.25in,clip,keepaspectratio]{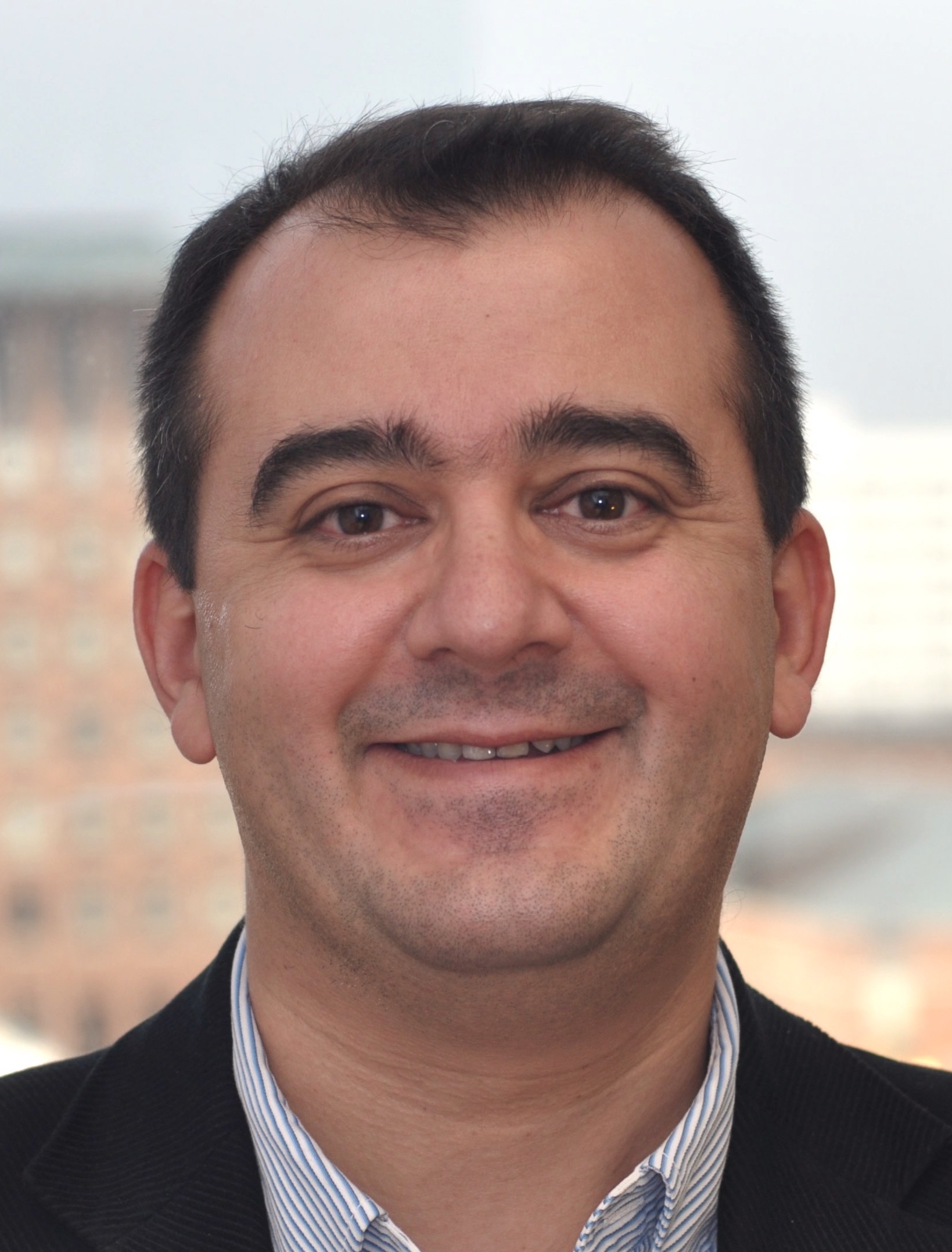}}]{Eylem Ekici}
(S’99-M’02-SM’11-F’17) received his B.S. and M.S. degrees in computer engineering from Bogazici University, Turkey, in 1997 and 1998, respectively, and his Ph.D. degree in electrical and computer engineering from the Georgia Institute of Technology in 2002. Currently, he is a Professor with the Department of Electrical and Computer Engineering, The Ohio State University. His current research interests are in the general area of wireless communication systems and networks, with a focus on algorithm design and resource management for mmWave, dynamic spectrum, and vehicular communication systems. He served as the general Co-Chair of ACM MobiCom 2012 and ACM MobiHoc 2021. He was also the TPC Co-Chair of the IEEE INFOCOM 2017. He is an Associate Editor-in-Chief for the IEEE TRANSACTIONS ON MOBILE COMPUTING, and a former Associate Editor for the IEEE/ACM TRANSACTIONS ON NETWORKING, the IEEE TRANSACTIONS ON MOBILE COMPUTING, and Elsevier Computer Networks.
\end{IEEEbiography}

\begin{IEEEbiography}
[{\includegraphics[width=1in,height=1.25in,clip,keepaspectratio]{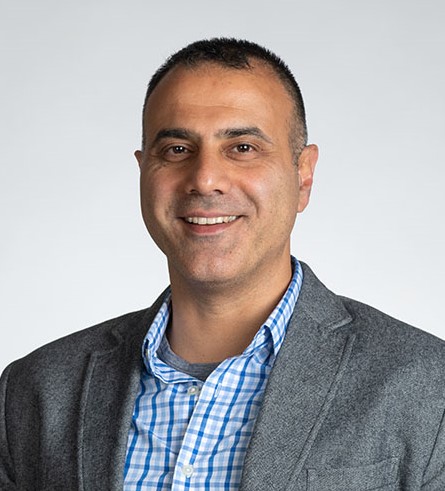}}]
{Can Emre Koksal}
(S’96–M’03–SM’13) received the B.S. degree in Electrical Engineering from the Middle East Technical University in 1996, and the S.M. and Ph.D. degrees from MIT in 1998 and 2002, respectively, in Electrical Engineering and Computer Science. He was a Postdoctoral Fellow at MIT until 2004, and a Senior Researcher at EPFL until 2006. Since then, he has been with the Electrical and Computer Engineering Department at Ohio State University, currently as a Professor. His general areas of interest are wireless communication, cybersecurity, communication networks, information theory, stochastic processes, and financial economics. He is also the Founder and CEO of DAtAnchor and several  technologies he has built have been licensed to a few enterprises and start-ups.

He is the recipient of Columbus Business First - Inventor of the Year Award in 2020, the National Science Foundation CAREER Award in 2011, a finalist of the Bell Labs Prize in 2015, OSU CoE Innovator Award in 2016 and 2020, OSU CoE Lumley Research Award in 2011 and 2017, and the co-recipient of an HP Labs - Innovation Research Award in 2011. Papers he co-authored received the best paper award in IEEE WiOpt 2018 and the best student paper candidate in ACM MOBICOM 2005. He has served as an Associate Editor for IEEE Transactions on Information Theory, IEEE Transactions on Wireless Communications, and Elsevier Computer Networks.
\end{IEEEbiography}
\vfill

\end{document}